\documentclass[twoside]{article}
\makeatletter
\let \@sverbatim \@verbatim
\def \@verbatim {\@sverbatim \verbatimplus}
{\catcode`'=13 \gdef \verbatimplus{\catcode`'=13 \chardef '=13 }} 
\makeatother

\usepackage{etoolbox}
\makeatletter
\patchcmd{\l@section}
  {\hfil}
  {\leaders\hbox{\normalfont$\m@th\mkern \@dotsep mu\hbox{.}\mkern \@dotsep mu$}\hfill}
  {}{}
\makeatother

\usepackage{tikz,stmaryrd,amssymb,amsmath,amsfonts,amsthm,url,color,relsize}
\usetikzlibrary{hobby,backgrounds,calc,trees}

\usepackage{textcomp}
\urldef{\oest}\url{http://www.gzbjzb.com/oeis.org/A065410}

\newcommand{\dl}{\textit{dl}}
\newcommand{\SigSCL}{\ensuremath{\Sigma_{\textup{SCL}}(A)}}
\newcommand{\SigCP}{\ensuremath{\Sigma_{\textup{CP}}(A)}}
\newcommand{\TCP}{\ensuremath{\mathbb{T}_{{\SigCP},\cal X}}}
\newcommand{\TSCL}{\ensuremath{\mathbb{T}_{{\SigSCL},\cal X}}}
\newcommand{\CPandneg}{\CP(\neg,\leftand,\leftor)}

\newcommand{\se}{\ensuremath{\textit{se}}}

\newcommand{\export}{\mathbin{\setlength{\unitlength}{1ex}
     \begin{picture}(2.0,1.8)(-.8,0)
     \put(-.5,1.6){\line(1,0){1.4}}
     \put(-.5,-0.2){\line(1,0){1.4}}
     \put(-.44,-0.2){\line(0,1){1.8}}
     \put(.84,-0.2){\line(0,1){1.8}}
     \end{picture}
     }}

\newcommand{\M}{\ensuremath{\mathbb{M}}}

\newcommand{\T}{\NT}

\newcommand{\axname}[1]{\textup{\ensuremath{\textrm{#1}}}}
\newcommand{\SCL}{\axname{SCL}}
\newcommand{\FSCL}{\axname{FSCL}}
\newcommand{\SCLe}{\axname{EqFSCL}}
\newcommand{\SCLi}{\axname{EqFSCL}^-}

\newcommand{\MSCLe}{\axname{EqMSCL}}
\newcommand{\MSCL}{\axname{MSCL}}
\newcommand{\SSCL}{\axname{SSCL}}
\newcommand{\SSCLe}{\axname{EqSSCL}}

\newcommand{\leftand}{~
     \mathbin{\setlength{\unitlength}{.9ex}
     \begin{picture}(1.6,1.8)(-.4,0)
     \put(-.8,0){\small$\wedge$}
     \put(-.66,-0.1){\textcolor{white}{\circle*{0.6}}}
     \put(-.66,-0.1){\circle{0.6}}
     \end{picture}
     }}
 
\newcommand{\fulland}{~
     \mathbin{\setlength{\unitlength}{.9ex}
     \begin{picture}(1.6,1.8)(-.4,0)
     \put(-.8,0){\small$\wedge$}
     \put(-.66,-0.1){\circle*{0.66}}
     \end{picture}
     }}
          
\newcommand{\leftor}{~
     \mathbin{\setlength{\unitlength}{.9ex}
     \begin{picture}(1.6,1.8)(-.4,0)
     \put(-.8,0){\small$\vee$}
     \put(-.66,1.3){\textcolor{white}{\circle*{0.6}}}
     \put(-.66,1.3){\circle{0.6}}
     \end{picture}
     }}

\newcommand{\PS}{\ensuremath{{\mathcal{C}_A}}}
\newcommand{\SP}{\ensuremath{{\mathcal{S}_A}}}
\newcommand{\NT}{\ensuremath{{\mathcal{T}_A}}}

\newcommand{\mem}{\ensuremath{\textit{mem}}}

\newcommand{\tr}{\ensuremath{{\sf T}}}
\newcommand{\fa}{\ensuremath{{\sf F}}}
\newcommand{\true}{\ensuremath{\textit{true}}}
\newcommand{\false}{\ensuremath{\textit{false}}}

\newcommand{\memt}{\ensuremath{\mathit{m}}}
\newcommand{\memse}{\ensuremath{\mathit{mse}}}

\newcommand{\Le}{L}
\newcommand{\Ri}{R}
\newcommand{\stse}{\ensuremath{\mathit{sse}}}
\newcommand{\Au}{A^{u}}

\newcommand{\SPf}[1]{\ensuremath{{\mathcal{S}_{#1}}}}
\newcommand{\PSf}[1]{\ensuremath{{\mathcal{C}_{#1}}}}

\newcommand{\lef}{\ensuremath{\scalebox{0.78}{\raisebox{.1pt}[0pt][0pt]{$\;\lhd\;$}}}}
\newcommand{\rig}{\ensuremath{\scalebox{0.78}{\raisebox{.1pt}[0pt][0pt]{$\;\rhd\;$}}}}
\renewcommand{\unlhd}{\ensuremath{\scalebox{0.78}{\raisebox{.1pt}[0pt][0pt]{$\;\trianglelefteq\;$}}}}
\renewcommand{\unrhd}{\ensuremath{\scalebox{0.78}{\raisebox{.1pt}[0pt][0pt]{$\;\trianglerighteq\;$}}}}

\newtheorem{theorem}{Theorem}[section]
\newtheorem{lemma}[theorem]{Lemma}
\newtheorem{proposition}[theorem]{Proposition}

\newtheorem{definition}[theorem]{Definition}
\newtheorem{fact}[theorem]{Fact}
\theoremstyle{definition}
\newtheorem{example}[theorem]{Example}

\newcommand{\CP}{\axname{CP}}
\newcommand{\CPmem}{\axname{\CP$_{\mem}$}}
\newcommand{\CPstat}{\axname{\CP$_{\mathit{s}}$}}

\newcommand{\qedex}{\textit{End~example.}}

\begin{document}

\title{Propositional logic with short-circuit evaluation: a non-commutative 
and a commutative variant}

\author{
	Jan A.\ Bergstra
	\qquad
	Alban Ponse \qquad 
	Daan J.C. Staudt\\[2mm]
  {\small 
  Section Theory of Computer Science, Informatics Institute}\\
 {\small  Faculty of Science, University of Amsterdam}\\[2mm]
	{\small \url{https://staff.science.uva.nl/{j.a.bergstra,a.ponse}}\quad\url{https://www.daanstaudt.nl}}
}

\date{}

\maketitle

\begin{abstract}
Short-circuit evaluation denotes the semantics of 
propositional connectives in which the second
argument is evaluated only if the first argument does not suffice 
to determine the value of the expression.
Short-circuit evaluation is widely used in programming, with
sequential conjunction and disjunction as primitive connectives.

We study the question which logical laws axiomatize short-circuit evaluation under the 
following assumptions:
compound statements are evaluated from left to right, 
each atom (propositional variable) evaluates to either true or false, 
and atomic evaluations can cause a side effect.
The answer to this question depends on the kind of atomic side effects that can occur
and leads to different ``short-circuit logics".
The basic case is FSCL (free short-circuit logic), which
characterizes the setting in which each atomic evaluation can cause a side effect. 
We recall some main results and then relate FSCL 
to MSCL (memorizing short-circuit logic), where
in the evaluation of a compound statement, the first evaluation result of each atom is memorized.
MSCL can be seen as a sequential variant of propositional logic:
atomic evaluations cannot cause a side effect and the sequential connectives are not 
commutative.
Then we relate MSCL to SSCL (static short-circuit logic), the variant of propositional 
logic that prescribes short-circuit evaluation with commutative sequential connectives.

We present evaluation trees as an intuitive semantics for short-circuit evaluation,
and simple equational axiomatizations for the short-circuit logics mentioned that
use negation and the sequential connectives only.
\\[3mm]
\end{abstract}
{\small\tableofcontents}

\section{Introduction}
In this paper, we discern a 
fixed evaluation strategy to determine the truth of a propositional statement.
We proceed from some very simple points of departure:
\begin{itemize}
\item
Atoms (propositional variables) evaluate to either \true\ or \false, thus we exclude logics that 
comprise other truth values.
\item
The semantics of the binary propositional connectives (conjunction and disjunction) is determined
by \emph{short-circuit evaluation}:
the second argument is evaluated only if the first argument does not suffice to determine 
the (evaluation) value of the expression. 
\item
Once an atom in a compound
expression is evaluated to a truth value, each next atomic evaluation of that atom
evaluates to the same truth
value. For example, if $a$ evaluates to \true, then so does $a\wedge a$.
\end{itemize}

We consider conjunction as the primary connective and disjunction as a derived connective, and we
write 
\[\leftand \text{ and }~\leftor\]
for the case that these connectives \emph{prescribe} short-circuit evaluation.
This notation stems from~\cite{BBR95}, where
the small circle indicates that the left argument must be evaluated first. 
Other notations are \texttt{\&\&} and \texttt{||} as used in programming,
$\otimes$ and $\oplus$ from transaction logic (see, e.g.~\cite{BK15}), and
$\vartriangle$ and $\triangledown$ from computability logic (see, e.g.~\cite{Jap08}).
However, we prefer the asymmetric symbols and we will henceforth refer to these as
\emph{sequential connectives}.
Given a set of atoms (propositional variables), sequential propositions are built from atoms, 
sequential conjunction and disjunction as mentioned here, negation, and the
constants \tr\ and \fa\ for the values \true\ and \false.

Short-circuit evaluation combines well with negation, and sequential (equational) variants of 
De Morgan's laws are valid, such as
\[\neg(x\leftand y)=\neg x\leftor \neg y.\]

We first recall \emph{free short-circuit logic}, \FSCL\ for short, 
and relate this to two variants of propositional logic with short-circuit evaluation.
In \FSCL, two sequential propositions are identified if and only if
they always have the same evaluation value under short-circuit evaluation. 
Here ``always'' refers to any possible
assumption about the truth value of atoms in any evaluation state, 
{and} to the side effects that may occur in the evaluation process: 
we speak of an \emph{atomic side effect} if the evaluation of an atom in a compound expression changes
(influences) the evaluation result of the subsequent atoms that must be evaluated to 
determine value of the expression. \FSCL\ is a logic for equational reasoning about sequential
propositions that may have atomic side effects without any restriction. 
Stated differently, this logic is immune to all atomic side effects.
For example, in \FSCL\ the sequential proposition $a\leftand a$ is not equivalent with $a$ or with
$a\leftand(a\leftor a)$.
Two typical laws of \FSCL\ are $(x\leftand y)\leftand z=x\leftand (y\leftand z)$ and 
$x\leftand\fa=\neg x\leftand\fa$.

In this paper we study two short-circuit logics that comprise \FSCL:
\begin{description}

\item[$\MSCL,$] ``memorizing short-circuit logic", is a logic for equational reasoning about sequential
propositions with the property that atomic side effects do not occur: in the evaluation of 
a compound statement the first evaluation result of each atom is memorized. 
In this logic, the sequential connectives are not commutative, for example,
$a\leftand\fa$ and
$\fa\leftand a$ are not equivalent (the first sequential proposition requires evaluation
of atom $a$, the second one does not).
Typical laws of \MSCL\ are $x\leftand x=x$ and $x\leftand (y\leftand x)=x\leftand y$.

\item[$\SSCL,$] ``static short-circuit logic", is the (equational)
variant of propositional logic in which short-circuit evaluation is prescribed,
thus the sequential connectives are taken to be commutative.
Also this logic is based on the assumption that atomic side effects do not occur.
\end{description}

\paragraph{Structure of the paper and main results.}
In Section~\ref{sec:FSCL} we discuss evaluation trees, which model the evaluation of a 
sequential proposition and were defined in~\cite{Daan,PS17}.
We recall the main results on \FSCL, in particular its equational 
axiomatization \SCLe\ for closed terms.

In Section~\ref{sec:MSCL} we define memorizing evaluation trees by a transformation on the 
evaluation trees
for \FSCL,  introduce \MSCLe\ as an equational axiomatization of their equality, and
show that the axioms of \SCLe\ are derivable from \MSCLe. 

In Section~\ref{sec:Hoare} we recall the definitions of the short-circuit logics mentioned
above. These definitions employ  the \emph{conditional} |
a ternary connective introduced by Hoare in 1985 in~\cite{Hoa85} | 
as a hidden operator, and stem from~\cite{BPS13,BP12a}.

In Section~\ref{sec:Completeness} we prove that \MSCLe\ corresponds with \MSCL\ 
in the sense that both define the same equational theory, and that both axiomatize equality of 
memorizing evaluation trees.

In Section~\ref{sec:SSCL} we define \SSCLe\ as the extension of \MSCLe\ with a commutativity
axiom, and prove that \SSCLe\ is an equational axiomatization of \SSCL.
Then we show that both axiomatize equality of static evaluation trees as defined in~\cite{BP15}. 
Finally, we present four simple axioms for the conditional connective 
as an alternative for those in~\cite{Hoa85}.

Section~\ref{sec:Conc} contains some conclusions, in particular on viewing both \MSCL\ and \SSCL\ as
variants of propositional logic. 

\paragraph{Notes.}1.
All derivability results in this paper were checked with the theorem 
prover \emph{Prover9}, and all independence results were found with help of  
the tool \emph{Mace4}, see~\cite{BirdBrain} for both these tools. 
We added four appendices with detailed proofs of these results.

\noindent
2. Considerable parts of the text below stem from~\cite{Daan,BPS13,PS17}.
Together with~\cite{PS17}, this paper subsumes most of~\cite{BPS13}. 
Two topics discussed in~\cite{BPS13} and not in this paper are
`repetition-proof' and `contractive' short-circuit logic; we will deal with these topics 
in a forthcoming paper.

\section{Evaluation trees and axioms for short-circuit evaluation}
\label{sec:FSCL}
In this section we summarize the main results of~\cite{PS17}:
evaluation trees and an axiomatization of their equality are discussed 

\medskip
Given a non-empty set $A$ of atoms, we first define evaluation trees.

\begin{definition}
\label{def:treesN}
The set \NT\ of \textbf{evaluation trees} over $A$ with leaves in 
$\{\tr, \fa\}$ is defined inductively by
\[\tr\in\NT,\quad\fa\in\NT, \quad (X\unlhd a\unrhd Y)\in\NT 
~\text{ for any }X,Y \in \NT \text{  and } a\in A.\]
The operator $\_\unlhd a\unrhd\_$ is called 
\textbf{tree composition over $a$}.
In the evaluation tree $X \unlhd a \unrhd Y$, 
the root is represented by $a$,
the left branch by $X$ and the right branch by $Y$. 
\end{definition}
The leaves of an evaluation tree represent evaluation results (so we use 
the constants \tr\ and \fa\ for \true\ and \false). 
Next to the formal notation for evaluation
trees we also use a more pictorial representation. For example,
the tree
\[\fa\unlhd b\unrhd(\tr\unlhd a\unrhd\fa)\]
can be represented as follows, where $\unlhd$ yields a left branch, and $\unrhd$ a right branch:
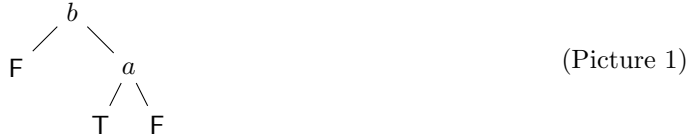
\begin{equation}
\label{plaatje}
\tag{Picture 1}
\hspace{-12mm}
\begin{tikzpicture}[%
      level distance=7.5mm,
      level 1/.style={sibling distance=15mm},
      level 2/.style={sibling distance=7.5mm},
      baseline=(current bounding box.center)]
      \node (a) {$b$}
        child {node (b1) {$\fa$}
        }
        child {node (b2) {$a$}
          child {node (d1) {$\tr$}} 
          child {node (d2) {$\fa$}}
        };
      \end{tikzpicture}
\end{equation}

In order to define a short-circuit semantics for negation and the sequential 
connectives, we first define the \emph{leaf replacement} operator, 
`replacement' for short, on trees in \NT\ as follows. 
For $X\in\NT$, the replacement of \tr\ with $Y$ and $\fa$ with $Z$ in $X$, denoted
\[X[\tr\mapsto Y, \fa \mapsto Z]\]
is defined recursively by 
\begin{align*}
\tr[\tr\mapsto Y,\fa\mapsto Z]&= Y,\\
\fa[\tr\mapsto Y,\fa\mapsto Z]&= Z,\\
(X_1\unlhd a\unrhd X_2)[\tr\mapsto Y,\fa\mapsto Z]
&=X_1[\tr\mapsto Y,\fa\mapsto Z]\unlhd a\unrhd X_2[\tr\mapsto Y,\fa\mapsto Z].
\end{align*}
We note that the order in which the replacements of leaves of 
$X$ is listed
is irrelevant and adopt the convention of not listing  
identities inside the brackets, e.g., 
$X[\fa\mapsto Z]=X[\tr\mapsto \tr,\fa\mapsto Z]$.
By structural induction it follows that repeated replacements satisfy 
\begin{align*}
X[\tr\mapsto Y_1,\fa\mapsto Z_1][\tr\mapsto Y_2,\fa\mapsto Z_2]=
X[\tr\mapsto Y_1[\tr\mapsto Y_2,\fa\mapsto Z_2],~
\fa\mapsto Z_1[\tr\mapsto Y_2,\fa\mapsto Z_2]].
\end{align*}

We define the set \SP\ of closed (sequential) propositional statements over $A$
by the following grammar:
\[P ::= a\mid\tr\mid\fa\mid \neg P\mid P\leftand P\mid P\leftor P,\]
where $a\in A$, \tr\ is a constant for the truth  value \true, \fa\ for \false,
and refer to its signature by
\[\SigSCL=\{\leftand,\leftor,\neg,\tr,\fa,a\mid a\in A\}.\]
We interpret propositional statements in \SP\ as evaluation trees
by a function $se$ (abbreviating  short-circuit evaluation).

\begin{definition}
\label{def:se}
The unary \textbf{short-circuit evaluation function} $se : \SP \to\NT$ 
is defined as
follows, where $a\in A$:
\begin{align*}
se(\tr) &= \tr,
&se(\neg P)&=se(P)[\tr\mapsto \fa,\fa\mapsto \tr],\\
se(\fa) &= \fa,
&se(P \leftand Q)&= se(P)[\tr\mapsto se(Q)],\\
se(a)&=\tr\unlhd a\unrhd \fa,
&se(P \leftor Q)&= se(P)[\fa\mapsto se(Q)].
\end{align*}
\end{definition}

The overloading of the symbol \tr\ in $se(\tr)=\tr$ will not
cause confusion (and similarly for \fa).
As a simple example we derive the evaluation tree for $\neg b\leftand a$:
\[se(\neg b\leftand a)=se(\neg b)[\tr\mapsto se(a)]=
(\fa\unlhd b\unrhd\tr)[\tr\mapsto se(a)]=\fa\unlhd b\unrhd(\tr\unlhd a\unrhd\fa),\] 
which can be visualized as~\ref{plaatje} on page~\pageref{plaatje}.
Also, $se(\neg(b\leftor\neg a))=\fa\unlhd b\unrhd(\tr\unlhd a\unrhd\fa)$.
An evaluation tree $se(P)$ represents short-circuit evaluation in a way that can be
compared to the notion of a truth table for propositional logic in that it 
represents each possible evaluation of $P$. However, there are some important differences with
truth tables: in $se(P)$, the sequentiality
of $P$'s evaluation is represented, and
the same atom may occur multiple times in $se(P)$.

\begin{definition}
\label{def:freevc}
The binary relation \textbf{$se$-congruence}, notation $=_\se$, is defined on \SP\ by 
\[P=_\se Q~\iff~ se(P)=se(Q).\]
\end{definition}

\begin{table}
\centering
\hrule
\begin{align}
\label{SCL1}
\tag{F1}
\fa&=\neg\tr\\[0mm]
\label{SCL2}
\tag{F2}
x\leftor y&=\neg(\neg x\leftand\neg y)\\[0mm]
\label{SCL3}
\tag{F3}
\neg\neg x&=x\\[0mm]
\label{SCL4}
\tag{F4}
\tr\leftand x&=x\\[0mm]
\label{SCL5}
\tag{F5}
x\leftor\fa&=x\\[0mm]
\label{SCL6}
\tag{F6}
\fa\leftand x&=\fa\\[0mm]
\label{SCL7}
\tag{F7}
(x\leftand y)\leftand z&=x\leftand (y\leftand z)
\\[0mm]
\label{SCL8}
\tag{F8}
\qquad
\neg x\leftand \fa&= x \leftand\fa
\\[0mm]
\label{SCL9}
\tag{F9}
(x\leftand\fa)\leftor y
&=(x\leftor\tr)\leftand y\\
\label{SCL10}
\tag{F10}
(x\leftand y)\leftor(z\leftand\fa)&=
(x\leftor (z\leftand\fa))\leftand(y\leftor (z\leftand\fa))
\end{align}
\hrule
\caption{\SCLe, a set of axioms for $se$-congruence} 
\label{tab:SCL}
\end{table}

In~\cite{Daan,PS17} it is proved that the axioms in Table~\ref{tab:SCL}\footnote{%
  In~\cite{Daan}, the dual of axiom~\eqref{SCL5} is used.}   
constitute an equational axiomatization of $se$-congruence:
\begin{fact}
\label{fact:fscl}
For all $P,Q\in\SP$,
$\SCLe\vdash P=Q~\iff~P=_\se Q$.
\end{fact}
This implies that the axioms in Table~\ref{tab:SCL} axiomatize 
free short-circuit logic \FSCL\ (defined in
Section~\ref{sec:Hoare}) for closed terms,
and for this reason this set of axioms is named $\SCLe$.
Some comments on these axioms: \eqref{SCL1}-\eqref{SCL3} imply sequential versions of 
De~Morgan's laws, and thus a sequential variant of
the duality principle.
Axioms~\eqref{SCL4}-\eqref{SCL6} define how the constants \tr\ and \fa\
interact with the sequential connectives, and axiom~\eqref{SCL7} defines
the associativity of $\leftand$. 

Axiom~\eqref{SCL8} defines 
a typical property of a logic that characterizes immunity for side effects: 
although it is the case that for each $P\in\SP$,
the evaluation result of $P\leftand\fa$ is
\false, the evaluation of $P$ might also yield a side effect. 
However, the same side effect and evaluation result
are obtained upon evaluation of $\neg P\leftand\fa$.

Axiom~\eqref{SCL9} expresses another property that concerns
possible side effects: because the 
evaluation result of $P\leftand\fa$ for each possible evaluation  of 
the atoms in $P$ is \false, $Q$ is always evaluated in $(P\leftand\fa)\leftor Q$
and determines the evaluation result, which is also the case in $(P\leftor\tr)\leftand Q$.
Note that the evaluations of $P\leftor\tr$ and $P\leftand\fa$ accumulate the same side effects,
which perhaps is more easily seen if one replaces $Q$ by either \tr\ or \fa.

Axiom~\eqref{SCL10} defines a restricted form of 
right-distributivity of ${\leftor}$ and (by duality) of ${\leftand}$.
This axiom holds because if $x$ evaluates to \true, both sides further evaluate
$y\leftor (z\leftand\fa)$, and if $x$ evaluates to \false,
$z\leftand\fa$ determines the further evaluation result (which is then \false,
and by axiom~\eqref{SCL6},  $y\leftor (z\leftand \fa)$ is not evaluated in the right-hand side).

The dual of $P\in\SP$, notation $P^{\dl}$, is defined as follows (for $a\in A$):
\begin{align*}
\tr^{\dl}&=\fa,
&a^{\dl}&=a,
&(P\leftand Q)^{\dl}&= P^{\dl}\leftor Q^{\dl},\\
\fa^{\dl}&=\tr,
&(\neg P)^{\dl}&=\neg P^{\dl},
&(P\leftor Q)^{\dl}&= P^{\dl}\leftand Q^{\dl}.
\end{align*} 
The duality mapping $(\:)^{\dl}:\SP\to\SP$ is an involution, that is, $(P^{\dl})^{\dl}=P$.
Setting $x^{\dl}=x$ for each variable $x$, the duality principle
extends to equations, e.g.,
the dual of axiom~\eqref{SCL7} is $(x\leftor y)\leftor z = x\leftor (y\leftor z)$.
From~\eqref{SCL1}-\eqref{SCL3} it immediately follows that \SCLe\ satisfies the duality principle, 
that is, for all terms $s,t$ over \SigSCL,
\[\SCLe\vdash s=t\quad\iff\quad\SCLe\vdash s^{\dl}=t^{\dl}.\]

We conclude this section with some more properties of $\SCLe$ that were proved in~\cite{PS17}.
\begin{fact}
\label{thm:cfscl}
Let $\SCLi=\SCLe\setminus\{\eqref{SCL1},\eqref{SCL3}\}$. Then
\[\text{$\SCLi\setminus\{\eqref{SCL8},\eqref{SCL10}\}
\vdash\eqref{SCL1},\eqref{SCL3}$, and thus $\SCLi\vdash\SCLe$,}\]
and the axioms of $\SCLi$ are independent if $A$ contains at least two atoms.
\end{fact}

\section{Evaluation trees and axioms for memorizing short-circuit evaluation}
\label{sec:MSCL}
In this section memorizing evaluation trees and an axiomatization of their equality are introduced,
as well as a congruence that identifies more than $se$-congruence.

\medskip

A short-circuit evaluation is \emph{memorizing} if 
in the evaluation of 
a compound statement the first evaluation result of each atom is memorized.
Typically, the following sequential version of the absorption law holds under memorizing evaluations:
\begin{align}
\label{MSCL1}
\tag{Abs}
x\leftand (x \leftor y)
&=x.
\end{align}
Equation~\eqref{MSCL1} can be explained as follows: 
if $x$ evaluates to \false, then $x\leftand (x \leftor y)$ evaluates to \false\ as a result of the evaluation of the left occurrence of $x$ (and $(x\leftor y)$
is not evaluated);
if $x$ evaluates to \true,  the second evaluation of
$x$ in the subterm $(x\leftor y)$ also results in \true\ (because it is 
memorizing) and therefore $y$ is not evaluated.

A perhaps less obvious property of memorizing evaluations is the following:
\begin{align}
\label{Mem}
\tag{Mem}
(x\leftor y)\leftand z
&=(\neg x\leftand (y\leftand z))\leftor (x\leftand z).
\end{align}
If $x$ evaluates to \true, then $z$ 
determines the evaluation result of both expressions because the evaluation 
result of $x$ is memorized;
if $x$ evaluates to \false, the evaluation result of both expressions is 
determined by $y\leftand z$ because the right disjunct $(x\leftand z)$
also evaluates to \false. 

Below we define the \emph{memorizing evaluation function} as
a transformation on evaluation trees.
This transformation implements the characteristic of memorizing evaluations
starting at the root of an evaluation tree, and removes each second occurrence
of a label $a$ according to its first evaluation result. 
Intuitively, memorizing evaluations are those of propositional logic,
except that the sequential connectives are not commutative. As an example, 
$a\leftand b$ and
$b\leftand a$ represent different evaluations, and hence are not equivalent.
\begin{definition}
\label{def:memse}
The unary 
\textbf{memorizing evaluation function} 
\[\memse:\SP\to \T\]
yields \textbf{memorizing evaluation trees} and is defined by
\begin{align*}
\memse(P)&=\memt(se(P)).
\end{align*}
The auxiliary function $\memt:\T\to\T$ is defined as follows ($a\in A$):
\begin{align*}
\memt(\tr)&=\tr,\\
\memt(\fa)&=\fa,\\
\memt(X\unlhd a\unrhd Y)&=\memt(\Le_a(X))\unlhd a\unrhd\memt(\Ri_a(Y)).
\end{align*}
For $a\in A$, the auxiliary functions $\Le_a: \T\to\T$ (``Left $a$-reduction'')
and $\Ri_a: \T\to\T$ (``Right $a$-reduction'')
are defined by
\begin{align*}
\Le_a(\tr)&=\tr,\\
\Le_a(\fa)&=\fa,\\
\Le_a(X\unlhd b\unrhd Y)&=
\begin{cases}
\Le_a(X)
&\text{if $b=a$},\\
\Le_a(X)\unlhd b\unrhd \Le_a(Y)
&\text{otherwise},
\end{cases}
\end{align*}
and
\begin{align*}
\Ri_a(\tr)&=\tr,\\
\Ri_a(\fa)&=\fa,\\
\Ri_a(X\unlhd b\unrhd Y)&=
\begin{cases}
\Ri_a(Y)
&\text{if }b= a,\\
\Ri_a(X)\unlhd b\unrhd \Ri_a(Y)
&\text{otherwise}.
\end{cases}
\end{align*}
\end{definition}

As an example we depict $se(a\leftand( b\leftand a))$
and the memorizing evaluation tree $\memse(a\leftand( b\leftand a))$:
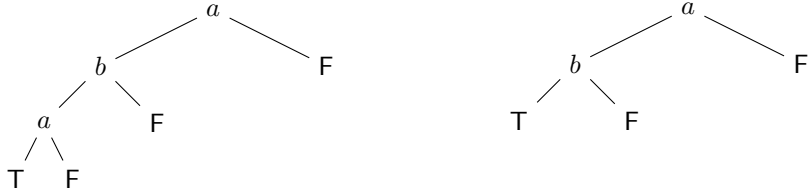
\[
\begin{array}{ll}
\begin{array}{l}
\begin{tikzpicture}[%
level distance=7.5mm,
level 1/.style={sibling distance=30mm},
level 2/.style={sibling distance=15mm},
level 3/.style={sibling distance=7.5mm}
]
\node (a) {$a$}
  child {node (b1) {$b$}
    child {node (c1) {$a$}
      child {node (d1) {$\tr$}} 
      child {node (d2) {$\fa$}}
    }
    child {node (c2) {$\fa$}
    }
  }
  child {node (b2) {$\fa$}
  };
\end{tikzpicture}
\end{array}
&\qquad
\begin{array}{l}
\qquad
\begin{tikzpicture}[%
level distance=7.5mm,
level 1/.style={sibling distance=30mm},
level 2/.style={sibling distance=15mm},
level 3/.style={sibling distance=7.5mm}
]
\node (a) {$a$}
  child {node (b1) {$b$}
    child {node (c1) {$\tr$}
    }
    child {node (c2) {$\fa$}
    }
  }
  child {node (b2) {$\fa$}
  };
\end{tikzpicture}
\\[8mm]
\end{array}\end{array}
\]

From a more general point of view, a memorizing evaluation tree is a  
\emph{decision tree}, that is
a labeled, rooted, binary tree with internal nodes labeled from 
$A$ and leaves labeled from $\{\tr,\fa\}$ such that
for any path from the root to a leaf, the internal nodes receive distinct labels (cf.~\cite{Moret}). 

Equality of memorizing evaluation trees defines a congruence on \SP.

\begin{definition}
\label{def:MSCL}
\textbf{Memorizing $se$-congruence}, notation $=_\memse^{\cal S}$,
is defined on \SP\ by
\[P=_\memse^{\cal S} Q ~\iff~ \memse(P)=\memse(Q).\]
\end{definition}
The superscript ${\cal S}$ in $=_\memse^{\cal S}$ is used as a reference to \SP\
because later on we will consider a close variant of this congruence.
In Section~\ref{sec:Completeness} we argue why $=_\memse^{\cal S}$ is a congruence.
Memorizing $se$-congruence identifies much more
than $se$-congruence, but not as much as propositional logic, e.g.,
${\leftand}$ and ${\leftor}$ are not commutative: $\fa\leftand a\ne_\memse^{\cal S} a\leftand\fa$.

\medskip

In Table~\ref{tab:MSCL} we present
a set of equational axioms for $=_\memse^{\cal S}$ and we call this set \MSCLe\ (this is a simplified
version of \MSCLe\ as introduced in~\cite{BPS13,BP12a}). One of our main results, proved in 
Section~\ref{sec:Completeness}, is the following:
\[
\tag{\text{Thm~\ref{thm:mscl2}}}
\text{For all $P,Q\in\SP$, $\MSCLe\vdash P=Q~\iff~ P=_\memse^{\cal S} Q$.}
\]
To enhance readability, we renamed the \SCLe-axioms used: 
$\eqref{SCL1}\to\eqref{Neg}$, $\eqref{SCL2}\to\eqref{Or}$, and $\eqref{SCL4}\to\eqref{Tand}$.
  
\begin{table}
\centering
\hrule
\begin{align}
\label{Neg}
\tag{Neg}
\fa&=\neg\tr
\\
\label{Or}
\tag{Or}
x\leftor y&=\neg(\neg x\leftand\neg y)
\\ 
\label{Tand}
\tag{Tand}
\tr\leftand x&=x
\\ 
\tag{\ref{MSCL1}}
x\leftand(x\leftor y)&=x\\
\label{MSCL3}
\tag{\ref{Mem}}
(x\leftor y)\leftand z
&=(\neg x\leftand (y\leftand z))\leftor (x\leftand z)
\end{align}
\hrule
\caption{$\MSCLe$, a set of axioms for memorizing $se$-congruence}
\label{tab:MSCL}
\end{table}

\begin{theorem}
\label{thm:Indep}
The axioms of $\MSCLe$ are independent.
\end{theorem}
\begin{proof} 
By Theorem~\ref{thm:Indep2} (that states that a superset of 
 $\MSCLe$ is independent).
\end{proof}
The next theorem states that the \SCLe-axioms are derivable from \MSCLe,
and  hence implies that $se$-congruence is subsumed by memorizing $se$-congruence.

\begin{theorem}
\label{thm:extraM}
$\MSCLe\vdash\SCLe$.
\end{theorem}
\begin{proof}
With help of the theorem prover \emph{Prover9}, see Appendix~\ref{app:mscl}.
\end{proof} 

In the proof of Theorem~\ref{thm:extraM}, the \SCLe-axioms are derived in a particular order,
as to obtain useful intermediate results.
Axiom~\eqref{SCL3}, that is $\neg\neg x=x$, is derived first, which justifies the use of 
the duality principle in subsequent derivations.
For easy reference we mention here some particular results 
that are used in Section~\ref{sec:Completeness}.

\begin{fact}
\label{fact:mscl}
The following equations are derivable from \MSCLe\ (and proved in Appendix~\ref{app:mscl}).

\noindent
Axioms~\eqref{MSCL1} and~\eqref{SCL5}, that is $x\leftor\fa=x$ (easy to derive), 
imply \textbf{idempotence} of ${\leftand}$:
\[
x=x\leftand(x\leftor\fa)=x\leftand x.\]
Two auxiliary results (with simple proofs) that are repeatedly used 
are the following:
\begin{align*}
\tag{\ref{ar 2}}
&x\leftand y=(x\leftand\fa)\leftor (x\leftand y),\\
\tag{\ref{ar 1}}
&x\leftor y=(\neg x\leftand y)\leftor x.
\end{align*}
The following two intermediate results are used in the derivation of axiom~\eqref{SCL10}
and in Section~\ref{sec:Hoare}
(note that with memorizing evaluations, these terms all express 
``\textup{\texttt{if $x$ then $y$ else $z$}}''):
\begin{align}
\tag{\ref{M1}}
&(x\leftand y)\leftor(\neg x\leftand z)
=(\neg x\leftor y)\leftand(x\leftor z),
\\
\tag{\ref{M2}}
&(x\leftand y)\leftor(\neg x\leftand z)=(\neg x\leftand z)\leftor(x\leftand y).
\\\nonumber
\end{align}
\end{fact}

A typical \MSCLe-consequence is $x\leftand (y\leftand x)=x\leftand y$ (cf.\ the last
\label{typical}
example on memorizing evaluation trees).
First derive 
\begin{equation}
\label{verbeterd}
\neg x\leftand \fa\stackrel{\eqref{SCL8}}=x\leftand \fa\stackrel{\eqref{ar 1}'}=(\neg x\leftor \fa)\leftand x\stackrel{\eqref{SCL5}}=
\neg x\leftand x,
\end{equation}
where \eqref{ar 1}$'$ is the dual of~\eqref{ar 1}.
Hence,
\begin{align*}
x\leftand y
&=(\neg x\leftor y)\leftand x
&&\text{by~\eqref{ar 1}$'$}\\
&=(x\leftand (y\leftand x))\leftor(\neg x\leftand x)
&&\text{by~\eqref{Mem}, \eqref{SCL3}}\\
&=(x\leftand (y\leftand x))\leftor(\neg x\leftand \fa)
&&\text{by~\eqref{verbeterd}}\\
&=(\neg x\leftand \fa)\leftor (x\leftand(y\leftand x))
&&\text{by~\eqref{M2}}\\
&=x\leftand(y\leftand x).
&&\text{by~\eqref{SCL8}, \eqref{ar 2}}
\end{align*}

Another convenient result on \MSCLe, used in Section~\ref{sec:Completeness}, is the following.

\begin{theorem}
\label{thm:dist}
The following equations are derivable from \MSCLe, where~\eqref{D1} abbreviates left-distributivity
of $\leftand$.
\begin{align}
\label{M3}
\tag{M3}
((x\leftand y)\leftor(\neg x\leftand z))\leftand u
&=(x\leftand (y\leftand u))\leftor(\neg x\leftand (z\leftand u)),
\\
\label{D1}
\tag{LD}
x\leftand(y\leftor z)
&=(x\leftand y)\leftor(x\leftand z).
\end{align}
\end{theorem}

\begin{proof} 
With help of the theorem prover \emph{Prover9}, see Appendix~\ref{app:LD}.
\end{proof}

We end this section by mentioning two alternatives for \MSCLe.
\begin{proposition}
\label{prop:X}
Replacing axiom~\eqref{Mem} in \MSCLe\ by~\eqref{M1}, and either~\eqref{M3} or
\[((x\leftand y)\leftor(\neg x\leftand z))\leftand u
=(\neg x\leftand (z\leftand u))\leftor(x\leftand (y\leftand u))\]
(thus, \eqref{M3}'s commutative variant) constitutes an alternative for \MSCLe.
\end{proposition}

Both these sets of axioms are independent (by~\emph{Mace4}~\cite{BirdBrain}). 
With~\emph{Prover9}~\cite{BirdBrain}, derivations of~\eqref{ar 2}, \eqref{ar 1}, 
\eqref{M2},
and \eqref{Mem}, respectively, are simple. 
In contrast to the proof of Theorem~\ref{thm:extraM}, a
derivation of the associativity of $\leftand$ (axiom~\eqref{SCL7}) 
is so simple that we show it here:
\begin{align*}
(x\leftand y)\leftand z
&=((\neg x\leftand\fa)\leftor(x\leftand y))\leftand z
&&\text{by~\eqref{ar 2}, \eqref{SCL8}}\\
&=(\neg x\leftand (\fa\leftand z))\leftor (x\leftand(y\leftand z))
&&\text{by~\eqref{M3}, \eqref{SCL3}}\\
&=((x\leftand\fa)\leftor (x\leftand(y\leftand z)))
&&\text{by~\eqref{SCL6}, \eqref{SCL8}}\\
&=x\leftand(y\leftand z).
&&\text{by~\eqref{ar 2}}
\end{align*}

\section{The conditional connective and three short-circuit logics}
\label{sec:Hoare}
In this section we consider Hoare's \emph{conditional}, a ternary connective that can be used for 
defining the sequential connectives of $\SigSCL=\{\leftand,\leftor,\neg,\tr,\fa,a\mid a\in A\}$.
Then we recall the definitions of free short-circuit logic (\FSCL),  
memorizing short-circuit logic (\MSCL), and static short-circuit logic (\SSCL)
that were published earlier.

\medskip
In 1985, Hoare introduced the \emph{conditional} (\cite{Hoa85}), a
ternary connective with notation
\[x\lef y \rig z.\]
A more common expression for the conditional $x\lef y\rig z$
is ``$\texttt{if }y\texttt{ then }x\texttt{ else }z$'',
which emphasizes that $y$ is evaluated \emph{first}, and depending
on the outcome of this partial evaluation, either $x$ or $z$ is evaluated,
which then determines the evaluation result. So, the evaluation strategy
prescribed by this form of if-then-else is a prime example of a sequential
evaluation strategy.
In order to reason algebraically with conditional expressions, 
Hoare's `operator like' notation $x\lef y\rig z$ seems indispensable.
In~\cite{Hoa85} an equational axiomatization of 
propositional logic is provided that only uses the conditional. Furthermore it is 
described how the sequential connectives and negation are expressed in this set-up,
although the sequential nature of the conditional's
evaluation is not discussed in this paper.
Hoare's axiomatization over the signature $\SigCP=\{\_\lef\_\rig\_\,,\tr,\fa,a\mid a\in A\}$
consists of eleven axioms,
including those in Table~\ref{CP}.
In Section~\ref{sec:SSCL} we present a concise and simple alternative for this axiomatization.

\begin{table}
\centering
\rule{1\textwidth}{.4pt}
\begin{align*}
\label{cp1}
\tag{CP1} x \lef \tr \rig y &= x\\
\label{cp2}\tag{CP2}
x \lef \fa \rig y &= y\\
\label{cp3}\tag{CP3}
\tr \lef x \rig \fa  &= x\\
\label{cp4}\tag{CP4}
\qquad
    x \lef (y \lef z \rig u)\rig v &= 
	(x \lef y \rig v) \lef z \rig (x \lef u \rig v)
\end{align*}
\hrule
\caption{The set \CP\ of axioms for proposition algebra}
\label{CP}
\end{table}

We extend the definition of the function $se$ 
(Definition~\ref{def:se}) to closed terms over $\SigCP$ by adding the clause
\begin{equation}
\label{def:extr}
se(P\lef Q\rig R)=se(Q)[\tr\mapsto se(P),\fa\mapsto se(R)].
\end{equation}
The four axioms in Table~\ref{CP}, named \CP\ (for Conditional Propositions),
establish a complete axiomatization of $se$-congruence over the signature \SigCP: 
\begin{equation*}
\text{For all closed terms $P,Q$ over \SigCP, $\CP\vdash P=Q~\iff~ se(P)=se(Q)$.}
\end{equation*}
A simple proof of this fact is recorded in~\cite[Thm.2.11]{BP15} (and repeated in~\cite{PS17}).

With the conditional connective and the constants \tr\ and \fa, 
the sequential connectives prescribing short-circuit evaluation are definable:
\begin{align}
\label{defneg}
\neg x &=\fa\lef x\rig \tr,\\
\label{defand}
x\leftand y &=y\lef x\rig\fa,\\
\label{defor}
x\leftor y &=\tr\lef x\rig y.
\end{align}
Note that these equations agree with the extension of the definition of 
the function $se$ in~\eqref{def:extr} above:
$se(\neg P)=se(\fa\lef P\rig\tr)$, $se(P\leftand Q)=se(Q\lef P\rig\fa)$, and 
$se(P\leftor Q)=se(\tr\lef P\rig Q)$.
Thus, the axioms in Table~\ref{CP} combined with equations~\eqref{defneg}-\eqref{defor}, say 
\[\CPandneg,\]
axiomatize equality of evaluation trees for closed terms over the enriched signature
$\SigCP\cup\SigSCL$.

\medskip

In order to capture memorizing evaluations, the following axiom is formulated in~\cite{BP10}:
\begin{align}
\label{CPmem}
\tag{CPmem}
x\lef y\rig(z\lef u\rig(v\lef y\rig w))
&= x\lef y\rig(z\lef u\rig w)
\end{align}
The axiom~\eqref{CPmem} expresses that the first evaluation value of $y$ is memorized.
We define 
\[\CPmem=\CP \cup\{\eqref{CPmem}\}. \]
In forthcoming proofs we use the fact that replacing the variable $y$ in axiom~\eqref{CPmem} by 
$\fa\lef y\rig\tr$ and/or the variable $u$ by $\fa\lef u\rig\tr$ yields equivalent versions of 
this axiom:
\begin{align}
\label{CPmem1}
\tag{CPmem1}
\qquad
(x\lef y\rig(z\lef u\rig v))\lef u\rig w&=
(x\lef y\rig z)\lef u\rig w,\\
\label{CPmem2}
\tag{CPmem2}
x\lef y\rig((z\lef y\rig u)\lef v\rig w)
&= x\lef y\rig(u\lef v\rig w),\\
\label{CPmem3}
\tag{CPmem3}
((x\lef y\rig z)\lef u\rig v)\lef y\rig w
&= (x\lef u\rig v)\lef y\rig w.
\end{align}
This follows easily with~\eqref{cp4}, \eqref{cp2}, \eqref{cp1}.
Furthermore, if we replace $u$ by $\fa$ in~\eqref{CPmem}, we find the \emph{contraction law}
\begin{equation}
\label{CPcon1}
\tag{CPcon1}
\qquad
x\lef y\rig(v\lef y\rig w)=x\lef y\rig w,
\end{equation} 
and replacing $u$ by $\tr$ in axiom~\eqref{CPmem3}
yields the symmetric contraction law 
\begin{equation}
\label{CPcon2}
\tag{CPcon2}
\qquad
(x\lef y\rig z)\lef y\rig w= x\lef y\rig w.
\end{equation}
With help of the tool \emph{Mace4}~\cite{BirdBrain} it easily follows 
that the axioms of \CPmem\ are independent, and therefore those of \CP\ are
also independent.

We write $\CPmem(\neg,\leftand,\leftor)$
for the axioms of $\CPmem$ extended with equations \eqref{defneg}-\eqref{defor}.
An important property of $\CPmem(\neg,\leftand,\leftor)$ is that the conditional connective
can be expressed with the sequential connectives and negation.
First, observe that it is trivial to derive 
\[\neg x\leftand z=\fa\lef x\rig z,\]
and hence
\begin{align}
\nonumber
(x\leftand y)\leftor(\neg x\leftand z)
&=\tr\lef(y\lef x\rig\fa)\rig(\fa\lef x\rig z)
&&\text{by~\eqref{defneg}-\eqref{defor} and the above}\\
\nonumber
&=(\tr\lef y\rig(\fa\lef x\rig z))\lef x\rig
(\fa\lef x\rig z)
&&\text{by \eqref{cp4}, \eqref{cp2}}\\
\nonumber
&=(\tr\lef y\rig\fa)\lef x\rig z
&&\text{by \eqref{CPmem1}, \eqref{CPcon1}}\\
\label{eq:indef}
&=y\lef x\rig z.
&&\text{by \eqref{cp3}}
\end{align}
In some cases it is convenient to use other equations: 
\begin{align}
\label{eq:indef2}
(\neg x\leftand z)\leftor(x\leftand y)
&=y\lef x\rig z,\\
\label{eq:indef3}
(x\leftor z)\leftand(\neg x\leftor y)
&=y\lef x\rig z,\\
\label{eq:indef4}
(\neg x\leftor y)\leftand(x\leftor z)
&=y\lef x\rig z,
\end{align}
which can all be proved from $\CPmem(\neg,\leftand,\leftor)$ in a similar way.

\medskip

In~\cite{BP12a,BPS13} a set-up is provided for defining short-circuit 
logics in a generic way with help of the
conditional by restricting the consequences of some \CP-axiomati\-zation extended with
equation~\eqref{defneg} (that is, $\neg x=\fa\lef x\rig\tr$) and equation~\eqref{defand} 
(i.e., $x\leftand y=y\lef x\rig\fa$) to the signature $\SigSCL$. 
So, the conditional connective is considered a \emph{hidden operator}.

The definition below uses the export operator ${\export}$ of \emph{Module algebra}~\cite{BHK90} 
to express this in a concise way:
in module algebra, $S\export X$  is the operation that 
exports the signature $S$ from module $X$ while declaring 
other signature elements hidden. 

\begin{definition}
\label{def:SCL}
A \textbf{short-circuit logic}
is a logic that implies the consequences
of the module expression
\begin{align*}
\SCL=\{\tr,\neg,\leftand\}\export(&\CP\cup\{\eqref{defneg},\eqref{defand}\}).
\end{align*}
\end{definition}

As a first example, $~\SCL\vdash \neg\neg x=x$~ 
can be proved as follows:
\begin{align}
\nonumber
\neg \neg x&=\fa\lef(\fa\lef x\rig\tr)\rig\tr
&&\text{by \eqref{defneg}}\\
\nonumber
&=(\fa\lef\fa\rig\tr)\lef x\rig(\fa\lef\tr\rig\tr)
&&\text{by~\eqref{cp4}}\\
\nonumber
&=\tr\lef x\rig\fa
&&\text{by~\eqref{cp2}, \eqref{cp1}}\\
\label{eq:dns}
&=x.
&&\text{by~\eqref{cp3}}
\end{align}

In~\cite{BP12a,BPS13}, the following short-circuit logics were defined:

\begin{definition}
\label{def:FSCL}
\textbf{Free short-circuit logic $(\FSCL)$}
is the short-circuit logic that implies no other 
consequences than those of the module expression \SCL.

\noindent
\textbf{Memorizing short-circuit logic $(\MSCL)$}
is the short-circuit logic that implies no other 
consequences than those of the module expression 
\[
\{\tr,\neg,\leftand\}\export(\CP\cup\{\eqref{defneg},\eqref{defand},\eqref{CPmem}\}).
\]
\textbf{Static short-circuit logic $(\SSCL)$}
is the short-circuit logic that implies no other 
consequences than those of the module expression 
\[
\{\tr,\neg,\leftand\}\export(\CP\cup\{\eqref{defneg},\eqref{defand},\eqref{CPmem}\}
\cup\{\fa\lef x\rig\fa=\fa\}).
\]
\end{definition}
To enhance readability, we extend these short-circuit logics with the constant \fa\
and its defining equation~\eqref{Neg}, which is justified by the \SCL-derivation
\begin{align}
\nonumber
\fa&=\fa\lef \tr\rig\tr
&&\text{by~\eqref{cp1}}\\
\label{eq:scl1}
&=\neg\tr,
&&\text{by~\eqref{defneg}}
\end{align}
and with the connective $\leftor$ and its defining 
equation~\eqref{Or} (thus, $x\leftor y=\neg(\neg x\leftand\neg y)$) by admitting
equation~\eqref{defor} in \SCL-derivations,
that is, $x\leftor y=\tr\lef x\rig y$. This last extension is justified by 
\begin{align}
\nonumber
\neg(\neg x\leftand\neg y)
&=\fa\lef(\neg y\lef(\fa\lef x\rig\tr)\rig\fa)\rig\tr
&&\text{by~\eqref{defneg}, \eqref{defand}}\\
\nonumber
&=\fa\lef(\fa\lef x\rig\neg y)\rig\tr
&&\text{by~\eqref{cp4}, \eqref{cp2}, \eqref{cp1}}\\
\nonumber
&=(\fa\lef\fa\rig\tr)\lef x\rig(\fa\lef\neg y\rig\tr)
&&\text{by~\eqref{cp4}}\\
\label{eq:scl2}
&=\tr\lef x\rig y.
&&\text{by~\eqref{cp2}, \eqref{defneg}, \eqref{eq:dns}}
\end{align}

In~\cite{Daan,PS17} the following results are proved:
\[\text{For all $P,Q\in\SP$, $\FSCL\vdash P=Q ~\iff~ \SCLe\vdash P=Q~\iff~P=_\se Q$.}\]
In the remainder of the paper we will prove similar results for \MSCL\ and  
\SSCL\ .

\section{Completeness of \MSCLe}
\label{sec:Completeness}

In this section we prove that $\MSCLe$ and \MSCL\ are equally
strong, that is, both define the same equational theory. Furthermore, both constitute
a complete axiomatization of memorizing $se$-congruence.

\medskip
Given a signature $\Sigma$, we write 
\[\mathbb T_{\Sigma,\cal X}\]
for the set of open terms over 
$\Sigma$ with variables in $\cal X$ (typical elements of $\cal X$ are $x,y,z,u,v,w$).
\begin{definition}
\label{def:functions}
Define the following two functions between sets of open terms: 
\begin{description}
\item
[$f:\TSCL\to\TCP$] is defined by 
\begin{align*}
f(bl)&=bl~\text{ for }bl\in\{\tr,\fa\},
&f(\neg t)&=\fa\lef f(t)\rig\tr,\\
f(a)&=a~\text{ for }a\in A,
&f(t_1\leftand t_2)&=f(t_2)\lef f(t_1)\rig\fa,\\
f(x)&=x~\text{ for }x\in \cal X,
&f(t_1\leftor t_2)&=\tr\lef f(t_1)\rig f(t_2).
\hspace{32mm}
\end{align*}
\item
[$g:\TCP\to\TSCL$] is defined by 
\begin{align*}
g(bl)&=bl~\text{ for }bl\in\{\tr,\fa\},
&g(x)&=x~\text{ for }x\in \cal X,
\\
g(a)&=a~\text{ for }a\in A,
&g(t_1\lef t_2\rig t_3)&=(g(t_2)\leftand g(t_1))\leftor(\neg g(t_2)\leftand g(t_3)).
\end{align*}
\end{description}
\end{definition}

\begin{lemma}
\label{lemma1}
For all $t\in\TSCL$, $\CPmem(\neg,\leftand,\leftor)\vdash f(t)=t$.
\end{lemma}

\begin{proof}
By structural induction on $t$.
\end{proof}

\begin{lemma}
\label{lemma2}
For all $s,t\in\TCP$, $\CPmem(\neg,\leftand,\leftor)\vdash s=t~\Rightarrow~\CPmem\vdash s=t$.
\end{lemma}

\begin{proof}
In an equational proof of $\CPmem(\neg,\leftand,\leftor)\vdash s=t$, each occurrence of one of 
the equations~\eqref{defneg}, \eqref{defand}, and \eqref{defor} can be replaced by the corresponding
\TCP-identity. 
More precisely, any occurrence of $\neg x= \fa\lef x\rig \tr$ can be replaced by 
$\fa\lef x\rig \tr= \fa\lef x\rig \tr$, and similar for applications of 
\eqref{defand} and \eqref{defor}. 

Because $s$ and $t$ do not contain occurrences of $\neg,\leftand$, 
and $\leftor$, this yields an equational proof of $s=t$ in $\CPmem$.
\end{proof}

\begin{lemma}
\label{lemma33}
For all $s,t\in\TCP$, $\CPmem\vdash s=t~\Rightarrow~ \MSCLe\vdash g(s)=g(t)$.
\end{lemma}

\begin{proof}
The $g$-translation of each \CPmem-axiom is derivable in \MSCLe.

\noindent
{Axiom~\eqref{cp1}}. $g(x\lef\tr\rig y)=(\tr\leftand x)\leftor(\neg \tr\leftand y)=x=g(x)$.

\noindent
{Axiom~\eqref{cp2}}. $g(x\lef\fa\rig y)=(\fa\leftand x)\leftor (\neg\fa\leftand y)=
\fa\leftor y=y=g(y)$.

\noindent
{Axiom~\eqref{cp3}}. $g(\tr\lef x\rig\fa)= (x\leftand \tr)\leftor(\neg x\leftand\fa)=
x\leftor (\neg x\leftand\fa)\stackrel{\eqref{SCL8}}=x\leftor (x\leftand\fa)
\stackrel{\eqref{MSCL1}'}=x=g(x)$.

\noindent
{Axiom~\eqref{cp4}}. We write ``Assoc" for applications of associativity, and we use
use~\eqref{M1}, \eqref{M2}, \eqref{M3}, and~\eqref{D1} (see Fact~\ref{fact:mscl}).
\begin{align*}
g(&x\lef (y\lef z\rig u)\rig v)\\
&=(g(y\lef z\rig u)\leftand x)\leftor(\neg g(y\lef z\rig u)\leftand v)\\
&=([(z\leftand y)\leftor(\neg z\leftand u)]\leftand x)\leftor 
(\neg[(z\leftand y)\leftor (\neg z\leftand u)]\leftand v )\\
&=([(z\leftand y)\leftor(\neg z\leftand u)]\leftand x)\leftor 
([(\neg z\leftor\neg y)\leftand (z\leftor\neg u)]\leftand v )\\
&=([(z\leftand y)\leftor(\neg z\leftand u)]\leftand x)\leftor 
([(z\leftand\neg y)\leftor (\neg z\leftand\neg u)]\leftand v )
&&\text{by~\eqref{M1}}\\
&=[(z\leftand (y\leftand x))\leftor(\neg z\leftand (u\leftand x))]\leftor 
  [(z\leftand(\neg y\leftand v))\leftor (\neg z\leftand(\neg u\leftand v))],~
&&\text{by~\eqref{M3}}
\end{align*}
and
\begin{align*}
g(&(x\lef y\rig v)\lef z\rig(x\lef u\rig v))\\
&=(z\leftand g(x\lef y\rig v))\leftor(\neg z\leftand g(x\lef u\rig v))\\
&=(z\leftand ((y\leftand x)\leftor (\neg y\leftand v)))\leftor
(\neg z\leftand ((u\leftand x)\leftor (\neg u\leftand v)))\\
&=([(z\leftand (y\leftand x))\leftor (z\leftand(\neg y\leftand v))])\leftor
([(\neg z\leftand (u\leftand x))\leftor (\neg z\leftand(\neg u\leftand v))])
&&\text{by~\eqref{D1}}\\
&=(z\leftand (y\leftand x))\leftor[((z\leftand 
(\neg y\leftand v))
\leftor
(\neg z\leftand (u\leftand x)))
\leftor
(\neg z\leftand (\neg u\leftand v))]
&&\text{by~Assoc}\\
&=(z\leftand (y\leftand x))\leftor[(
(\neg z\leftand (u\leftand x))
\leftor
(z\leftand (\neg y\leftand v))
)
\leftor
(\neg z\leftand (\neg u\leftand v))]
&&\text{by~\eqref{M2}}\\
&=[(z\leftand (y\leftand x))\leftor
(\neg z\leftand (u\leftand x))]
\leftor
[(z\leftand (\neg y\leftand v))
\leftor
(\neg z\leftand (\neg u\leftand v))].
&&\text{by~Assoc}
\end{align*}

\noindent
{Axiom~\eqref{CPmem}}. As argued in Section~\ref{sec:MSCL},
it is sufficient to derive axiom~\eqref{CPmem1}, that is,
\[(w\lef y\rig(z\lef x\rig u))\lef 
x\rig v=(w\lef y\rig z)\lef x\rig v.\]
Derive 
\begin{align*}
g((w\lef y\rig(z\lef x\rig u))\lef x\rig v)
&=(x\leftand g(w\lef y\rig(z\lef x\rig u)))\leftor(\neg x\leftand v)
\\
&=(x\leftand M)\leftor(\neg x\leftand v),
\\[2mm]
g((w\lef y\rig z)\lef x\rig v)
&=(x\leftand g(w\lef y\rig z))\leftor(\neg x\leftand v)
&&\text{by~\eqref{eq:indef}}\\
&=(x\leftand N)\leftor(\neg x\leftand v),
\end{align*}
so it suffices to derive $x\leftand M=x\leftand N$. 
We use one auxiliary result and we write $(n)'$ for the dual version of equation~$(n)$.
\begin{align}
\nonumber
x\leftand\fa
&=(x\leftand\fa)\leftand y
&&\text{by~\eqref{SCL6}, Assoc}\\
\nonumber
&=((x\leftand\fa)\leftor\fa)\leftand y
&&\text{by~\eqref{SCL5}}\\
\nonumber
&=((\neg x\leftor(\fa\leftor\fa))\leftand(x\leftor\fa))\leftand y
&&\text{by~\eqref{Mem}$'$}\\
\nonumber
&=(\neg x\leftand x)\leftand y
&&\text{by~\eqref{SCL5}}\\
\label{hulpje}
&=(x\leftand\neg x)\leftand y.
&&\text{by~\eqref{M2}}
\end{align}
Hence,
\begin{align*}
x\leftand M
&=x\leftand g(w\lef y\rig(z\lef x\rig u))
\\
&=x\leftand ((y\leftand w)\leftor(\neg y\leftand ((x\leftand z)\leftor(\neg x\leftand u))))
\\
&=x\leftand((\neg y\leftor w)\leftand(y\leftor ((x\leftand z)\leftor(\neg x\leftand u))))
&&\text{by~\eqref{M1}}\\
&=x\leftand((y\leftor ((x\leftand z)\leftor(\neg x\leftand u)))\leftand(\neg y\leftor w))
&&\text{by~\eqref{M2}}\\
&=x\leftand((y\leftor ((\neg x\leftand u)\leftor(x\leftand z)))
\leftand(\neg y\leftor w))
&&\text{by~\eqref{M2}}\\
&=[x\leftand(y\leftor ((\neg x\leftand u)\leftor(x\leftand z)))]
\leftand(\neg y\leftor w)
&&\text{by~Assoc}
\\
&=[(x\leftand y)\leftor (x\leftand((\neg x\leftand u)\leftor(x\leftand z)))]
\leftand(\neg y\leftor w)
&&\text{by~\eqref{D1}}
\\
&=[(x\leftand y)\leftor ((x\leftand(\neg x\leftand u))\leftor(x\leftand(x\leftand z)))]
\leftand(\neg y\leftor w)
&&\text{by~\eqref{D1}}
\\
&=[(x\leftand y)\leftor (((x\leftand\neg x)\leftand u)\leftor(x\leftand z))]
\leftand(\neg y\leftor w)
&&\text{by~Assoc, idempotence}
\\
&=[(x\leftand y)\leftor ((x\leftand\fa)\leftor(x\leftand z))]
\leftand(\neg y\leftor w)
&&\text{by \eqref{hulpje}}
\\
&=[(x\leftand y)\leftor (x\leftand(\fa\leftor z))]
\leftand(\neg y\leftor w)
&&\text{by~\eqref{D1}}
\\
&=(x\leftand(y\leftor z))\leftand(\neg y\leftor w)
&& \text{by~\eqref{SCL4}$'$, \eqref{D1}}\\
&=x\leftand((y\leftor z)\leftand(\neg y\leftor w))
&& \text{by~Assoc}\\
&=x\leftand g(w\lef y\rig z)
&& \text{by~\eqref{eq:indef3}}\\
&=x\leftand N.
&&\qedhere
\end{align*}
\end{proof}

\begin{theorem}
\label{thm:corr1}
For all terms $s,t$ over $\SigSCL$, $\MSCLe\vdash s=t~\iff~\MSCL\vdash s=t$.
\end{theorem}

\begin{proof}
($\Rightarrow$) 
It suffices to derive the axioms of \MSCLe\ from \MSCL.

\noindent
{Axiom~\eqref{Neg}}. See~\eqref{eq:scl1}.

\noindent
{Axiom~\eqref{Or}}. This follows from~\eqref{eq:scl2}.

\noindent
{Axiom~\eqref{Tand}}. $\tr\leftand x=x\lef \tr\rig\fa=x$.

\noindent
{Axiom~\eqref{MSCL1}}. $\smash{x\leftand(x\leftor y)=(\tr\lef x\rig y)\lef x\rig\fa\stackrel{\eqref{CPcon2}}=
\tr\lef x\rig\fa=x}$.

\noindent
{Axiom~\eqref{Mem}}. Denote $(x\leftor y)\leftand z=(\neg x\leftand (y\leftand z))
\leftor (x\leftand z)$ by $L=R$. Then
\begin{align*}
L&=z\lef(\tr\lef x\rig y)\rig \fa
\\
&=z\lef x\rig(z\lef y\rig\fa),
&&\text{by \eqref{cp4}, \eqref{cp1}}
\\[2mm]
R
&=\tr\lef((z\lef y\rig\fa)\lef(\fa\lef x\rig\tr)\rig\fa)\rig (z\lef x\rig\fa)
\\
&=\tr\lef(\fa\lef x\rig(z\lef y\rig\fa))\rig (z\lef x\rig\fa)
&&\text{by \eqref{cp4}, \eqref{cp2}, \eqref{cp1}}\\
&=[z\lef x\rig\fa]\lef x\rig [\tr\lef (z\lef y\rig\fa)\rig(z\lef x\rig\fa)]
&&\text{by \eqref{cp4}, \eqref{cp2}}\\
&=z\lef x\rig (\tr\lef(z\lef y\rig\fa)\rig\fa)
&&\text{by~\eqref{CPcon2}, \eqref{CPmem2}}\\
&=z\lef x\rig (z\lef y\rig\fa).
&&\text{by~\eqref{cp3}}
\end{align*}

\noindent
($\Leftarrow$)
\vspace{-4mm}
\begin{align*}
\MSCL\vdash s=t
&~\Rightarrow~ \CPmem(\neg,\leftand,\leftor)\vdash s=t
&&\text{by definition}\\
&~\Rightarrow~ \CPmem(\neg,\leftand,\leftor)\vdash f(s)=f(t)
&&\text{by Lemma~\ref{lemma1}}\\
&~\Rightarrow~ \CPmem\vdash f(s)=f(t)
&&\text{by Lemma~\ref{lemma2}}\\
&~\Rightarrow~ \MSCLe\vdash g(f(s))=g(f(t)).
&&\text{by Lemma~\ref{lemma33}}
\end{align*}
Hence, it suffices to derive for all $t\in\TSCL$, $\MSCLe\vdash g(f(t))=t$.
This follows easily by structural induction, we only show the inductive case $t=t_1\leftor t_2$:
\[
g(f(t_1\leftor t_2))\stackrel{\text{IH}}=(t_1\leftand \tr)\leftor(\neg t_1\leftand t_2)
\stackrel{\eqref{M2},\eqref{SCL5}'}=(\neg t_1\leftand t_2)\leftor t_1
\stackrel{\eqref{ar 1}}=t_1\leftor t_2.
\]
\end{proof}

Now, let $\PS$ be the set of closed terms over $\SigCP$.

\begin{definition}
\label{def:cal}
The binary relation $=_\memse^{\cal C}$ on $\PS$, \textbf{memorizing valuation congruence}, 
is defined  by 
\[P=_\memse^{\cal C} Q~\iff~\memse^{\cal C}(P)=\memse^{\cal C}(Q),\]
where the function $\memse^{\cal C}:\PS\to\T$ is defined as in Definition~\ref{def:memse},
except that the function $se$ is now defined as in~\eqref{def:extr}, that is,
\[se(P\lef Q\rig R)=se(Q)[\tr\mapsto se(P),\fa\mapsto se(R)].\]
\end{definition}

This definition stems from~\cite[Def.5.12]{BP15}.
In~\cite[Thm.5.14]{BP15} we prove this completeness result:
\begin{equation}
\label{eq:BP15}
\text{For all $P,Q\in\PS$, $\CPmem\vdash P=Q~\iff~P=_\memse^{\cal C} Q$.}
\end{equation}
This result depends on a non-trivial proof of the fact 
that $=_\memse^{\cal C}$ is a congruence on $\PS$. 

\begin{lemma}
\label{lemma3}
For all $P,Q\in\SP$, $\MSCLe\vdash P=Q~\Rightarrow~ P=_\memse^{\cal S} Q$.
\end{lemma}

\begin{proof}
We first show that $=_\memse^{\cal S}$ is a congruence on $\SP$.
By structural induction, 
$se(P)=se(f(P))$ for all $P\in\SP$, 
where the function $f$ is defined in Definition~\ref{def:functions}. 
Hence, 
\begin{equation}
\label{h2}
\memse(P)=\memse^{\cal C}(f(P)).
\end{equation}
Assume $P_i=_\memse^{\cal S} P_i'$ for $i\in\{1,2\}$.
By~\eqref{h2},
$f(P_i)=_\memse^{\cal C} f(P_i')$, and because
$=_\memse^{\cal C}$ is a congruence on \PS, 
\[f(P_1\leftand P_2)=f(P_2)\lef f(P_1)\rig \fa=_\memse^{\cal C} f(P_2')\lef f(P_1')\rig \fa
= f(P_1'\leftand P_2'),\]
and thus $\memse^{\cal C}(f(P_1\leftand P_2))=\memse^{\cal C}(f(P_1'\leftand P_2'))$. 
By~\eqref{h2}, $\memse(P_1\leftand P_2)=
\memse(P_1'\leftand P_2')$, and thus
$P_1\leftand P_2=_\memse^{\cal S} P_1'\leftand P_2'$. The remaining cases follow in a similar way.

\medskip

Next, for all $P,Q\in\SP$, if $\MSCLe\vdash P=Q$ then $P=_\memse^{\cal S} Q$.
This follows from the facts that $=_\memse^{\cal S}$ is a congruence on $\SP$ and that
each closed instance of each axiom of \MSCLe\ satisfies $=_\memse^{\cal S}$.\footnote{%
  Without loss of generality it can be assumed that substitutions happen first in equational proofs 
  (see, e.g., \cite{Groote}).}
We only show this for axiom~\eqref{MSCL1}:
\begin{align*}
\memse(P\leftand(P\leftor Q))
&=\memse^{\cal C}(f(P\leftand(P\leftor Q)))
&&\text{by~\eqref{h2}}\\
&=\memse^{\cal C}((\tr\lef f(P)\rig f(Q))\lef f(P)\rig\fa)
&&\text{by definition of $f$}\\
&=\memse^{\cal C}(\tr\lef f(P)\rig\fa)
&&\text{by~\eqref{eq:BP15}, \eqref{CPcon2}}\\
&=\memse^{\cal C}(f(P))
&&\text{by~\eqref{eq:BP15}, \eqref{cp3}}\\
&=\memse(P).
&&\text{by~\eqref{h2}}
\end{align*}
\end{proof}

\begin{theorem}
\label{thm:mscl}
For all $P,Q\in\SP$, $\MSCL\vdash P=Q~\iff~ P=_\memse^{\cal S} Q$.
\end{theorem}

\begin{proof}
($\Rightarrow$)
If $\MSCL\vdash P=Q$, then by Theorem~\ref{thm:corr1}, $\MSCLe\vdash P=Q$, and by Lemma~\ref{lemma3}, 
$P=_\memse^{\cal S} Q$.

($\Leftarrow$)
If $P=_\memse^{\cal S} Q$, then by~\eqref{h2}, $f(P)=_\memse^{\cal C} f(Q)$. 
By~\eqref{eq:BP15}, $\CPmem\vdash f(P)=f(Q)$,
and thus $\CPmem(\neg,\leftand,\leftor)\vdash f(P)=f(Q)$,
and thus by Lemma~\ref{lemma1}, $\CPmem(\neg,\leftand,\leftor)\vdash P=Q$. 
By definition of \MSCL\ it follows that $\MSCL\vdash P=Q$.
\end{proof}

Theorem~\ref{thm:corr1} establishes that
\MSCL\ is axiomatized by the equational logic \MSCLe,
and Theorem~\ref{thm:mscl} establishes that \MSCL\ axiomatizes equality of memorizing evaluation 
trees.
Combining these results leads to a final theorem on this matter, which establishes that
``memorizing short-circuit logic'' as a concept is independent of the conditional connective,
with memorizing evaluation trees
at hand as a simple and natural semantics for representing memorizing short-circuit evaluations.
This is fully in line with Fact~\ref{fact:fscl} on ``free short-circuit logic''.

\begin{theorem}
\label{thm:mscl2}
For all $P,Q\in\SP$, $\MSCLe\vdash P=Q~\iff~ P=_\memse^{\cal S} Q$.
\end{theorem}

\section{Static short-circuit logic}
\label{sec:SSCL}
Static short-circuit logic covers the case in which 
the sequential connectives are taken to be commutative.
In this section we first discuss two axiomatizations,
one that is an extension of \MSCLe\ with a commutativity axiom~\eqref{C1}, and 
the one used in \SSCL's definition (Def.~\ref{def:FSCL}).
Then we discuss static evaluation trees and two completeness results.
Finally, we provide four simple \CP-equations that axiomatize static valuation congruence.

\begin{table}
\centering
\hrule
\begin{align}
\fa&=\neg\tr
\tag{\ref{Neg}}
\\
x\leftor y&=\neg(\neg x\leftand\neg y)
\tag{\ref{Or}}
\\[0mm]
\tag{\ref{Tand}}
\tr\leftand x&=x
\\ 
\tag{\ref{MSCL1}}
x\leftand(x\leftor y)&=x
\\[0mm]
\tag{\ref{Mem}}
(x\leftor y)\leftand z&=(\neg x\leftand (y\leftand z))\leftor(x\leftand z)
\\[0mm]
\label{C1}
\tag{Comm}
x\leftand y&=y\leftand x
\end{align}
\hrule
\caption{$\SSCLe$, a set of axioms for $\SSCL$}
\label{tab:SSCL}
\end{table}

\medskip

All axioms in Table~\ref{tab:SSCL} represent common laws for
propositional logic when forgetting the prescribed short-circuit evaluation, except
axiom~\eqref{Mem}.
We name this set of axioms \SSCLe, and first
prove some familiar laws without making use of axioms~\eqref{Neg} 
and~\eqref{Tand}, and thus without using the constants \tr\ and \fa. 

\begin{theorem}
\label{thm:assoc}
The four \SSCLe-axioms~\eqref{Or}, \eqref{MSCL1}, \eqref{Mem}, and~\eqref{C1}
imply idempotence and associativity of ${\leftand}$ and $\leftor$, 
the double negation shift $\neg\neg x=x$ (that is, axiom~\eqref{SCL3}), and  
the equations
\begin{align}
\label{Tdef}
\tag{Tdef}
(x\leftor\neg x)\leftand y&=y,\\
\nonumber
\tag{\ref{D1}}
x\leftand(y\leftor z)
&=(x\leftand y)\leftor(x\leftand z).
\end{align}
Furthermore, if $|A|\geq 2$, these four axioms are independent.
\end{theorem}

\begin{proof}
The mentioned derivabilities follow
with help of the theorem prover \emph{Prover9}, see Appendix~\ref{app:assoc}.
For independence, see the proof of Theorem~\ref{thm:Indep2}. 
\end{proof}

This result is relevant because in \SSCLe\ the constants \tr\ and \fa\ are 
redundant (by equation~\eqref{Tdef} and axiom~\eqref{C1},
$x\leftor \neg x=y\leftor\neg y$). Note that in the setting without these constants,
the duality principle is captured by 
axiom~\eqref{Or} and the double negation shift.

By definition of \SSCLe\ we have the following theorem.

\begin{theorem}
\label{thm:cons2}
$\SSCLe\vdash\MSCLe$.
\end{theorem}
Furthermore, we have the following result, which implies that the axioms of
\MSCLe\ are independent as well (cf.\ Theorem~\ref{thm:Indep}).
\begin{theorem}
\label{thm:Indep2}
The axioms of $\SSCLe$ are independent.
\end{theorem}
\begin{proof} 
With help of the tool \emph{Mace4}~\cite{BirdBrain}, see Appendix~\ref{app:Indep2}.
\end{proof}

We now return to static short-circuit logic \SSCL\ as defined in Definition~\ref{def:FSCL}.
In Table~\ref{tab:stat}, the \CP-axiom
\begin{align}
\label{CPstat}
\tag{CPs}
\fa\lef x\rig\fa
&= \fa
\end{align}
is added to \CPmem\ and the resulting set of axioms is denoted $\CPstat$. 
This set of axioms stems from~\cite{BPS13,BP12a}.
First, we formulate the analogue of Lemma~\ref{lemma2} and
establish a correspondence result for \SSCLe\ and \SSCL.

\begin{table}
\centering
\rule{1\textwidth}{.4pt}
\begin{align*}
\tag{\ref{cp1}} 
x \lef \tr \rig y &= x\\
\tag{\ref{cp2}}
x \lef \fa \rig y &= y\\
\tag{\ref{cp3}}
\tr \lef x \rig \fa  &= x\\
\tag{\ref{cp4}}
\qquad
    x \lef (y \lef z \rig u)\rig v &= 
	(x \lef y \rig v) \lef z \rig (x \lef u \rig v)
\\
\tag{\ref{CPmem}}
x\lef y\rig(z\lef u\rig(v\lef y\rig w))
&= x\lef y\rig(z\lef u\rig w)\\
\tag{\ref{CPstat}}
\fa\lef x\rig\fa
&= \fa
\end{align*}
\hrule
\caption{$\CPstat$, the set of \CP-axioms used in \SSCL's definition (Def.~\ref{def:SCL})}
\label{tab:stat}
\end{table}

\begin{lemma}
\label{lemma2s}
For all $s,t\in\TCP$, $\CPstat(\neg,\leftand,\leftor)\vdash s=t~\Rightarrow~\CPstat\vdash s=t$.
\end{lemma}

\begin{proof}
See the proof of Lemma~\ref{lemma2}.
\end{proof}

\begin{theorem}
\label{thm:corr2}
For all terms $s,t$ over $\SigSCL$, $\SSCLe\vdash s=t~\iff~\SSCL\vdash s=t$.
\end{theorem}

\begin{proof}
($\Rightarrow$)
It suffices to derive the axioms of \SSCLe\ from \SSCL, so by the proof of Theorem~\ref{thm:corr1}
we have to derive axiom~\eqref{C1}. First derive
\begin{align}
x\stackrel{\eqref{cp2}}=z\lef\fa\rig x\stackrel{\eqref{CPstat}}=z\lef(\fa\lef y\rig\fa)\rig x
\stackrel{\eqref{cp4}}=(z\lef\fa\rig x)\lef y\rig(z\lef\fa\rig x)
\stackrel{\eqref{cp2}}=x\lef y\rig x.
\label{eq:resu2}
\end{align}
Hence
\begin{align*}
x\leftand y
&=y\lef x\rig\fa\\
&=y\lef (x\lef y\rig x)\rig\fa
&&\text{by~\eqref{eq:resu2}}\\
&=((\tr\lef y\rig\fa)\lef x\rig\fa)\lef y\rig((\tr\lef y\rig\fa)\lef x\rig\fa)
&&\text{by~\eqref{cp4}, \eqref{cp1}}\\
&=(\tr\lef x\rig\fa)\lef y\rig(\fa\lef x\rig\fa)
&&\text{by~\eqref{CPmem3}, \eqref{CPmem2}}\\
&=x\lef y\rig\fa
&&\text{by~\eqref{cp1}, \eqref{CPstat}}\\
&=y\leftand x.
\end{align*}

\noindent
($\Leftarrow$)
Consider the functions $f$ and $g$ defined in Definition~\ref{def:functions}.
We extend Lemma~\ref{lemma33} to \CPstat:
\begin{equation}
\label{3stat}
\text{For all $s,t\in\TCP$, $\CPstat\vdash s=t~\Rightarrow~\SSCLe\vdash g(s)=g(t)$.}
\end{equation}
The additional proof obligation  is to show that the $g$-translation of the 
axiom~\eqref{CPstat} is derivable in \SSCLe\ (cf.\ Lemma~\ref{lemma33}):
\begin{align*}
g(\fa\lef x\rig\fa)
&=(x\leftand\fa)\leftor(\neg x\leftand\fa)\\
&=(\fa\leftand x)\leftor(\fa\leftand \neg x)
&&\text{by~\eqref{C1}}\\
&=\fa
&&\text{by~\eqref{SCL6}, \eqref{SCL5} (and Theorems~\ref{thm:cons2}, \ref{thm:extraM})}\\
&=g(\fa).
\end{align*} 
We adapt the $(\Leftarrow)$-part of the proof of Theorem~\ref{thm:corr1} to \SSCL. 
\begin{align*}
\SSCL\vdash s=t
&~\Rightarrow~ \CPstat(\neg,\leftand,\leftor)\vdash s=t
&&\text{by definition}\\
&~\Rightarrow~ \CPstat(\neg,\leftand,\leftor)\vdash f(s)=f(t)
&&\text{by Lemma~\ref{lemma1}}\\
&~\Rightarrow~ \CPstat\vdash f(s)=f(t)
&&\text{by Lemma~\ref{lemma2s}}\\
&~\Rightarrow~ \SSCLe\vdash g(f(s))=g(f(t)).
&&\text{by~\eqref{3stat}}
\end{align*}
Hence, it suffices to show for all $t\in\TSCL$, $\SSCLe\vdash g(f(t))=t$, and
by $\SSCLe\vdash\MSCLe$ (Thm.\ref{thm:cons2})
this follows as in the $(\Leftarrow)$-part of the proof of Theorem~\ref{thm:corr1}.
\end{proof}

\medskip

In~\cite{BP15}, static evaluation trees for conditional propositions
are defined with help of memorizing evaluation trees.
The crux is that given a conditonal proposition $P$ and a finite set of atoms
$A'$ that contains all atoms in $P$'s evaluation, the evaluation tree of $P$ is defined relative
to an ordering of $A'$. We denote such an ordering as a string of length $|A'|$
that covers $A'$, for example, the orderings of $A'=\{a,b\}$ are denoted by $ab$ and $ba$.
We write 
\[\Au\]
for the set of strings representing all such orderings, and 
\(\SPf{\sigma}\)
for the set of sequential propositions with atoms in $\sigma\in\Au$.
Before defining the static evaluation function, we give an example.

\begin{example}
\label{ex:st} Let $P=\neg a\leftor (b\leftand a)$. We depict 
$se(P)$ at the left-hand side, and
two static evaluation trees for $P$. 
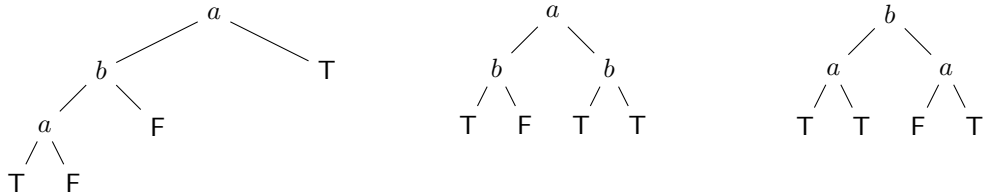
\[
\begin{array}{lll}
\begin{array}{l}
\begin{tikzpicture}[%
level distance=7.5mm,
level 1/.style={sibling distance=30mm},
level 2/.style={sibling distance=15mm},
level 3/.style={sibling distance=7.5mm}
]
\node (a) {$a$}
  child {node (b1) {$b$}
    child {node (c1) {$a$}
      child {node (d1) {$\tr$}} 
      child {node (d2) {$\fa$}}
    }
    child {node (c2) {$\fa$}
    }
  }
  child {node (b2) {$\tr$}
  };
\end{tikzpicture}
\end{array}
&\quad
\begin{array}{l}
\quad
\begin{tikzpicture}[%
level distance=7.5mm,
level 1/.style={sibling distance=15mm},
level 2/.style={sibling distance=7.5mm},
level 3/.style={sibling distance=7.5mm}
]
\node (a) {$a$}
  child {node (b1) {$b$}
    child {node (c1) {$\tr$}
    }
    child {node (c2) {$\fa$}
    }
  }
  child {node (b2) {$b$}
    child {node (c3) {$\tr$}
    }
    child {node (c4) {$\tr$}
    }
  };
\end{tikzpicture}
\\[8mm]
\end{array}
&\quad
\begin{array}{l}
\qquad
\begin{tikzpicture}[%
level distance=7.5mm,
level 1/.style={sibling distance=15mm},
level 2/.style={sibling distance=7.5mm},
level 3/.style={sibling distance=7.5mm}
]
\node (a) {$b$}
  child {node (b1) {$a$}
    child {node (c1) {$\tr$}
    }
    child {node (c2) {$\tr$}
    }
  }
  child {node (b2) {$a$}
    child {node (c3) {$\fa$}
    }
    child {node (c4) {$\tr$}
    }
  };
\end{tikzpicture}
\\[8mm]
\end{array}
\end{array}
\]
The two static evaluation trees correspond to the different ways in which one 
can present a (minimal) truth table for $P$, that is, the different possible
orderings of the valuation values of the atoms occurring in $P$:
\[
\renewcommand*{\arraystretch}{1.2}
\begin{array}{ll|c}
a&b&~\neg a\leftor (b\leftand a)~\\\hline
\tr&\tr~&\tr\\
\tr&\fa&\fa\\
\fa&\tr&\tr\\
\fa&\fa&\tr
\end{array}
\hspace{1cm}
\begin{array}{ll|c}
b&a&~\neg a\leftor (b\leftand a)~\\\hline
\tr&\tr~&\tr\\
\tr&\fa&\tr\\
\fa&\tr&\fa\\
\fa&\fa&\tr
\end{array}
\]

The idea is that each proposition with atoms in $\{a,b\}$ has a static evaluation tree
that is either of the form of the middle tree, or of the tree on the right,
depending on which $\sigma\in \Au$ is chosen,
and that the leaves represent the appropriate evaluation results. E.g., the leaves
in the static evaluation trees for \fa\ and for $\fa\leftand P$ are all \fa. 
\hfill
\qedex
\end{example}

Because static evaluation trees do not necessary reflect the order of atomic evaluations, we
do not take the trouble to define these directly for \SP, but reuse their definition for 
\PS, taken from~\cite[Def.6.13]{BP15}.\footnote{%
  We come back to this point in Section~\ref{sec:Conc}.}
For $\sigma\in\Au$, let
$\PSf{\sigma}$ be the set of closed terms over $\SigCP$ with atoms in $\sigma$.

\begin{definition}
\label{def:6.13}
Let $\sigma\in\Au$. The unary \textbf{static evaluation function} 
\[\stse_\sigma^{\cal C}:\PSf{\sigma}\to \T\]
yields \textbf{static evaluation trees} and is defined as follows:
\[\stse_\sigma^{\cal C}(P)=\memse^{\cal C}(\tr\lef E_\sigma\rig P),\]
with $\memse^{\cal C}$ as in Definition~\ref{def:cal},
and $E_\sigma$ defined by $E_{a\rho}=E_\rho\lef a\rig E_\rho$ if $\sigma=a\rho$ for $a\in A$,
and $E_\epsilon=\fa$ with $\epsilon$ the empty string.
\end{definition}

As an example, the static evalution tree
$\stse_{ab}^{\cal C}(\fa)=\stse_{ab}^{\cal C}(\fa\lef a\rig\fa)=
\stse_{ab}^{\cal C}(\fa\lef b\rig\fa)$ is depicted at the left-hand side, and
$\stse_{ba}^{\cal C}(\fa)=\stse_{ba}^{\cal C}(\fa\lef a\rig\fa)=
\stse_{ba}^{\cal C}(\fa\lef b\rig\fa)$ is the other tree.
\[
\begin{array}{lll}
\begin{array}{l}
\qquad
\begin{tikzpicture}[%
level distance=7.5mm,
level 1/.style={sibling distance=30mm},
level 2/.style={sibling distance=15mm},
level 3/.style={sibling distance=7.5mm}
]
\node (a) {$a$}
  child {node (b1) {$b$}
    child {node (c1) {$\fa$}
    }
    child {node (c2) {$\fa$}
    }
  }
  child {node (b2) {$b$}
    child {node (c3) {$\fa$}
    }
    child {node (c4) {$\fa$}
    }
  };
\end{tikzpicture}
\end{array}
&\qquad
\begin{array}{l}
\quad
\begin{tikzpicture}[%
level distance=7.5mm,
level 1/.style={sibling distance=30mm},
level 2/.style={sibling distance=15mm},
level 3/.style={sibling distance=7.5mm}
]
\node (a) {$b$}
  child {node (b1) {$a$}
    child {node (c1) {$\fa$}
    }
    child {node (c2) {$\fa$}
    }
  }
  child {node (b2) {$a$}
    child {node (c3) {$\fa$}
    }
    child {node (c4) {$\fa$}
    }
  };
\end{tikzpicture}
\end{array}
\end{array}
\]

Static evaluation trees are perfect binary trees, where each level characterises 
the evaluation of a single atom. 

\begin{definition}
\label{def:SSCL}
Let $\sigma\in A^u$. The binary relation $=_{\stse,\sigma}^{\cal C}$ on $\PSf{\sigma}$,
\textbf{static valuation congruence over $\sigma$}, 
is defined by
\[P=_{\stse,\sigma}^{\cal C} Q ~\iff~ \stse_\sigma^{\cal C}(P)=\stse_\sigma^{\cal C}(Q).\]
\end{definition}

This definition stems from~\cite[Def.6.14]{BP15}.
In~\cite[Thm.6.16]{BP15} we prove this completeness result:
\begin{equation}
\label{eq:BP152}
\text{Let $\sigma\in A^u$.
For all $P,Q\in\PSf{\sigma}$, $\CPstat\vdash P=Q~\iff~P=_{\stse,\sigma}^{\cal C} Q$.}
\end{equation}
This result depends on a non-trivial proof of the fact that $=_{\stse,\sigma}^{\cal C}$ 
is a congruence on $\PSf{\sigma}$. 
We define the following variants of static evaluation trees and static 
valuation congruence for $\SPf{\sigma}$.

\begin{definition}
\label{def:SSCL2}
Let $\sigma\in A^u$.
The unary \textbf{static evaluation function} $\stse_\sigma:\SPf{\sigma}\to\T$
is defined by $\stse_\sigma(P)=\stse_\sigma^{\cal C}(f(P))$, where $\stse_\sigma^{\cal C}$ 
and $f$ are defined in Definitions~\ref{def:6.13} and~\ref{def:functions}.

The binary relation $=_{\stse,\sigma}^{\cal S}$ on $\SPf{\sigma}$,
\textbf{static $se$-congruence over $\sigma$}, 
is defined on $\SPf{\sigma}$ by
\[P=_{\stse,\sigma}^{\cal S} Q ~\iff~ f(P)=_{\stse,\sigma}^{\cal C} f(Q).\]
\end{definition}

Hence, the two trees in the example above are also the static evaluation trees  
$\stse_{ab}(\fa)=\stse_{ab}(a\leftand\fa)=\stse_{ab}(b\leftand\fa)$ and
$\stse_{ba}(\fa)=\stse_{ba}(a\leftand\fa)=\stse_{ba}(b\leftand\fa)$, respectively. 

\begin{theorem}
\label{thm:sscl}
Let $\sigma\in A^u$.
For all $P,Q\in\SPf{\sigma}$, $\SSCL\vdash P=Q~\iff~ P=_{\stse,\sigma}^{\cal S} Q$.
\end{theorem}

\begin{proof}
By Lemma~\ref{lemma1},  it follows that 
for all $R\in\SPf{\sigma}$,
$\CPstat\cup\{\eqref{defneg},\eqref{defand},\eqref{defor}\}\vdash R=f(R)$.
Hence, $\SSCL\vdash P=Q\Longleftrightarrow\CPstat\cup\{\eqref{defneg},\eqref{defand},\eqref{defor}\}\vdash P=Q 
\iff\CPstat\cup\{\eqref{defneg},\eqref{defand},\eqref{defor}\}\vdash f(P)=f(Q)
\iff\CPstat\vdash f(P)=f(Q)$, where the last implication $\Rightarrow$ follows from Lemma~\ref{lemma2s}.
By~\eqref{eq:BP152}, the latter derivability holds if and only if
$f(P)=_\memse^{\cal C} f(Q)$, that is, $P=_\memse^{\cal S} Q$.
\end{proof}

It is cumbersome, but not difficult to define static evaluation trees directly from
memorizing evaluation trees: 
adapt Definition~\ref{def:SSCL2} by defining 
$D_{a\rho}=(a\leftand\neg a)\leftor D_\rho$, $D_\epsilon=\fa$,
$\stse_\sigma(P)=\memse(D_\sigma\leftor P)$, and
$P=_{\stse,\sigma}^{\cal S} Q \Longleftrightarrow \stse_\sigma(P)=\stse_\sigma(Q)$.
This defines exactly the same static evaluation trees and relation $=_{\stse,\sigma}^{\cal S}$, 
and thus provides
a semantics for static short-circuit evaluations without use of the conditional connective.
In this case, Theorem~\ref{thm:sscl} can be proved in a similar way as Theorem~\ref{thm:mscl}
(which would then require the analogue of Lemma~\ref{lemma3}).

By Theorem~\ref{thm:corr2},
static short-circuit logic (\SSCL) is axiomatized by the equational logic \SSCLe.
By Theorem~\ref{thm:sscl}, \SSCL\ axiomatizes equality of static evaluation trees.
Thus, ``static short-circuit logic'' as a concept is independent of the conditional connective,
and leads to the following completeness theorem.

\begin{theorem}
\label{thm:sscl2}
Let $\sigma\in\Au$.
For all $P,Q\in\SPf{\sigma}$, $\SSCLe\vdash P=Q~\iff~ P=_{\memse,\sigma}^{\cal S} Q$.
\end{theorem}

Observe that this is again fully in line with Fact~\ref{fact:fscl} on ``free short-circuit logic''
and Theorem~\ref{thm:mscl2} on ``memorizing short-circuit logic''.

\begin{table}
\centering
\hrule
\begin{align*}
\tag{\ref{cp1}}
x \lef \tr \rig y &= x\\
\tag{\ref{cp2}}
x \lef \fa \rig y &= y
\\ 
\label{CPcomb}\tag{CP3s}
\qquad	(x\lef y\rig z)\lef y\rig \fa
&=y\lef x\rig\fa
\\ 
\tag{\ref{cp4}}
x \lef (y \lef z \rig u)\rig v &= 
(x \lef y \rig v) \lef z \rig (x \lef u \rig v)
\end{align*}
\hrule
\caption{An alternative set of \CP-axioms for defining \SSCL}
\label{tab:stat2}
\end{table}

We conclude this section with a few words on the definition of 
static short-circuit logic (Def.~\ref{def:SCL}).
In Table~\ref{tab:stat2} we provide an alternative set of axioms for defining \SSCL,
thus for defining static valuation congruence.
This axiomatization is independent (which easily follows with \emph{Mace4}~\cite{BirdBrain}), 
but is not a simple extension of \CP\ or \CPmem.
Note that the axiom~\eqref{CPcomb} with $y=\tr$ implies~\eqref{cp3},
and with $y=\fa$ the axiom~\eqref{CPstat}.
A proof of one of the axioms~\eqref{CPmem1} or~\eqref{CPmem3} by~\emph{Prover9}~\cite{BirdBrain} is relatively simple (with the option kbo);
for the first one, a convenient intermediate result is
\[\texttt{f(f(f(x,y,z),u,v),y,0)=f(f(x,u,v),y,0)},\]
that is,
\[((x\lef y\rig z)\lef u\rig v)\lef y\rig\fa=(x\lef u\rig v)\lef y\rig\fa,\]
and adding this as a fifth axiom yields a comprehensible proof of~\eqref{CPmem1}.

However, finding a more simple axiomatization of static valuation congruence is not a 
purpose of this paper: the axiomatization $\CPstat$ in Table~\ref{tab:stat} is 
sufficiently simple and expresses the fundamental intuitions in an appropriate way.
Reasons to present the axiomatization in Table~\ref{tab:stat2}
are its independence (contrary to \CPstat, see below) and, of course,
its striking simplicity (cf.~\cite{Hoa85}).

\begin{proposition}
\label{thm:stat}
$\CPstat \setminus\{\eqref{cp1}\}\vdash\eqref{cp1}$, and
the axioms of $\CPstat \setminus\{\eqref{cp1}\}$ are independent.
\end{proposition}

\begin{proof}
First derive
\begin{align}
\nonumber
x\lef y\rig(z\lef u\rig y)
&=x\lef y\rig(z\lef u\rig (\tr\lef y\rig\fa))
&&\text{by~\eqref{cp3}}\\
\label{eq:resu1}
&=x\lef y\rig(z\lef u\rig \fa),
&&\text{by~\eqref{CPmem}}
\\[2mm]
\nonumber
x\lef\tr\rig y
&=x\lef\tr\rig (\tr\lef y\rig \fa)
&&\text{by~\eqref{cp3}}\\
\nonumber
&=x\lef\tr\rig (\tr\lef y\rig\tr)
&&\text{by~\eqref{eq:resu1}}\\
\label{eq:resu3}
&=x\lef\tr\rig \tr.
&&\text{by~\eqref{eq:resu2}}
\end{align}
Hence,
\begin{align*}
x\lef\tr\rig y
&=x\lef\tr\rig x
&&\text{by~\eqref{eq:resu3}}\\
&=x.
&&\text{by~\eqref{eq:resu2}}
\end{align*}

The independence of $\CPstat \setminus\{\eqref{cp1}\}$
follows easily with help of the tool \emph{Mace4}~\cite{BirdBrain},
where one atom is needed to show the independence of axiom~\eqref{cp4}
(recall that $A\ne\emptyset$).
\end{proof}

\section{Conclusions}
\label{sec:Conc}

In~\cite{BP10} we introduced `proposition 
algebra', which is based on Hoare's conditional $x\lef y\rig z$
and the constants $\tr$ and $\fa$.
We defined a number of varieties of so-called
\emph{valuation algebras} in order to capture different semantics for the 
evaluation of conditional statements, and provided axiomatizations for
the resulting valuation congruences: 
$\CP$ (four axioms) characterizes the least identifying valuation congruence
we consider, and the extension $\CPmem$
(one extra axiom) characterizes the
most identifying valuation congruence below ``sequential propositional
logic''. Static valuation congruence can be axiomatized by adding the  
axiom $\fa\lef x\rig\fa=\fa$
to $\CPmem$, and can be seen as a characterization of (sequential) propositional logic.

In~\cite{HMA,BP12a} we introduced an alternative valuation semantics
for proposition algebra in the form of \emph{Hoare-McCarthy algebras} (HMA's)
that is more elegant than the semantical framework
provided in~\cite{BP10}: HMA-based semantics
has the advantage that one can define a valuation congruence
without first defining the 
valuation \emph{equivalence} it is contained in. 

In~\cite{BP15}, following the approach of Staudt in~\cite{Daan}, we defined evaluation trees
as a more simple and direct semantics for proposition algebra and proved several 
completeness results for the valuation congruences mentioned above.

In~\cite{BPS13} we introduced ``short-circuit logic'' as defined here (Def.~\ref{def:SCL}
and Def.~\ref{def:FSCL}).
In~\cite{PS17}, we dealt with the case of free short-circuit logic (\FSCL), as 
is summarized in Section~\ref{sec:FSCL}.

\medskip

In this paper we establish a setting in which memorizing short-circuit logic \MSCL\
and static short-circuit logic \SSCL\ can be understood and used without any reference to 
(or dependence on) the conditional connective. 

From this perspective, \MSCL\ can be seen as the equational logic 
defined by \MSCLe\ and with equality of memorizing evaluation trees as a simple semantics.
\MSCL\ can also be viewed as a short-circuited, operational variant of propositional logic: 
decisive for the meaning of a sequential proposition is the \emph{process} of 
its sequential evaluation, as is clearly demonstrated by its memorizing evaluation tree,
which also explains why the sequential connectives are taken to be non-commutative and why 
the constants $\tr$ and \fa\ are not definable (and thus included). 
It is important to realize that a number of familiar properties hold in \MSCL:
\begin{itemize}
\item
The duality principle, the double negation shift, and associativity of the sequential
connectives (all of these hold in \FSCL).
\item 
Idempotence of the sequential connectives, and
\begin{align*}
x\leftand (y\leftand x)&=x\leftand y,
&&\text{(see page~\pageref{typical})}
\\x\leftand(y\leftor z) &=(x\leftand y)\leftor(x\leftand z).
&&\text{left-distributivity~\eqref{D1}}
\end{align*}
(none of these hold in \FSCL).
\end{itemize}
Some perhaps less familiar properties of \MSCL, none of which hold in \FSCL,
are the following
two characterizations of ``\texttt{if $x$ then $y$ else $z$}'':
\begin{align*}
(x\leftand y)\leftor(\neg x\leftand z)&=(\neg x\leftor y)\leftand(x\leftor z),
&&\text{\eqref{M1}}\\
(x\leftand y)\leftor(\neg x\leftand z)&=(\neg x\leftand z)\leftor(x\leftand y),
&&\text{\eqref{M2}}
\end{align*}   
and the right-distributivity of ${\leftand}$ over ``\texttt{if $x$ then $y$ else $z$}'',
that is
\[(\texttt{if $x$ then $y$ else $z$})\leftand u=\texttt{if $x$ then $(y\leftand u)$ else $(z\leftand u)$},\]
which is characterized by
\begin{align*}
((x\leftand y)\leftor(\neg x\leftand z))\leftand u
&=(x\leftand (y\leftand u))\leftor(\neg x\leftand (z\leftand u)).
&&\text{\eqref{M3}}
\end{align*}   
Also, \eqref{M1} and the dual of~\eqref{M3} imply  right-distributivity of ${\leftor}$ 
over ``\texttt{if $x$ then $y$ else $z$}'':
\begin{align*}
((x\leftand y)\leftor(\neg x\leftand z))\leftor u
&=(x\leftand (y\leftor u))\leftor(\neg x\leftand (z\leftor u)).
\end{align*}   

Likewise, we can view \SSCL\ as the equational logic 
defined by \SSCLe, and with equality of static evaluation trees as its semantics.\footnote{%
  There is an innocent difference between
  the definition of static evaluation trees used in this paper
  (Def.~\ref{def:6.13}) and its origin
  \cite[Def.6.13]{BP15}: the $\sigma$'s in the current definition are reversed, which we view as 
  more natural.}
However, it is questionable whether equality of static evaluation trees is a useful
semantics for \SSCL\ (or \SSCLe), despite
the interest of short-circuit connectives and short-circuit evaluation 
in propositional logic. 
Consider for example the identity
$a\leftand b=b\leftand a$,
which implies that the associated static evaluation trees should be considered equal.
So, this either requires a transformation of $se$-evaluation trees according to an
ordering of a fixed set of atoms (that contains $a$ and $b$),
which may not agree with the evaluation order of atoms, or a non-intuitive equivalence 
relation between (ordinary) evaluation trees that does not respect this evaluation order.
The same problem occurs in the case of expressions with the conditional and their static 
evaluation trees: the mismatch is that 
$b\lef a\rig\fa$ models a sequential, short-circuited
evaluation of $a\wedge b$, while the (necessary) identification
$b\lef a\rig\fa=a\lef b\rig\fa$  
declares the sequential nature of this evaluation irrelevant. 

We conclude with some comments on the differences between \MSCL\ and \SSCL.
First, the constant
$\tr$ is not definable in \MSCL, but in \SSCL\ it is definable  by $x\leftor \neg x$
(cf.\ Theorem~\ref{thm:assoc}). Next, short-circuit evaluation and \emph{full evaluation} 
(prescribed by ${\fulland}$, see~\cite{PS17,Daan}) 
do not coincide in \MSCL, but they do in \SSCL:
\[x\fulland y\stackrel{\textit{def}}=(x\leftor(y\leftand\fa))\leftand y=(x\leftor(\fa\leftand y))\leftand y=x\leftand y.\]
Furthermore, in both \MSCL\ and \SSCL, the number of semantically different formulas
is bounded by a function on $|A|$. This is an essential difference with short-circuit logics
that identify less, such as \FSCL.
For $|A|=n$ (recall $n>0$), the number of memorizing evaluation trees is 
$T_n = n(T_{n-1})^2 + 2$ with $T_0=2$ (so the first few
values are $6, 74, 16430$),\footnote{%
  See \oest.} 
and
for $\sigma=a_1a_2...a_n\in\Au$, the number of static evaluation trees over $\sigma$
is $2^{(2^n)}$.
We finally note that the complexity of deciding satisfiability for 
both \MSCL\ and \SSCL\ is NP-complete (see~\cite{Veld,BP10}).
All in all, taking short-circuit evaluation and the absence of atomic side effects
as points of departure, we think that \MSCL\ provides  
a more natural view on (sequential) propositional logic than \SSCL\ does.
\\[2mm]
\textbf{Related work.} 
In this paper we focused on the intrinsic properties of the sequential connectives in the setting of 
memorizing and static short-circuit evaluation, and we have not yet any specific applications in mind.
Nevertheless, we mention a few  areas of potentially related research. 
First, decision trees on Boolean variables as discussed in for example~\cite{Moret}
are memorizing evaluation trees.
Secondly, other notations for the sequential connectives $\leftand$ and $\leftor$ with memorizing 
interpretation
are $\vartriangle$ and $\triangledown$ from computability logic (see, e.g.~\cite{Jap08}),
and $\otimes$ and $\oplus$ from transaction logic (see, e.g.~\cite{BK15}), there called
\emph{serial} connectives. However,  \MSCL\ is just a part of both these 
logics and it is questionable whether its axiomatization or semantics are of any relevance.

\noindent
\textbf{Future work / Challenging questions.} 
With respect to the proof of Theorem~\ref{thm:extraM}, that is, $\MSCLe\vdash\SCLe$,
find a shorter and more comprehensible proof of associativity.
Alternatively,
find another equational axiomatization for \MSCL\ that is short and simple, 
uses only three variables, and admits a simple proof of this theorem.

\appendix
\renewcommand{\leftand}{~
     \mathbin{\setlength{\unitlength}{.9ex}
     \begin{picture}(1.6,1.8)(-.4,0)
     \put(-.8,0){\small$\wedge$}
     \put(-.65,-0.1){\textcolor{white}{\circle*{0.6}}}
     \put(-.65,-0.1){\circle{0.6}}
     \end{picture}
     }}
\renewcommand{\leftor}{~
     \mathbin{\setlength{\unitlength}{.9ex}
     \begin{picture}(1.6,1.8)(-.4,0)
     \put(-.8,0){\small$\vee$}
     \put(-.64,1.5){\textcolor{white}{\circle*{0.6}}}
     \put(-.64,1.55){\circle{0.6}}
     \end{picture}
     }}

\small

\section{Detailed proofs}

\subsection{A proof of Theorem~\ref{thm:extraM}}
\label{app:mscl}

\textbf{Theorem~\ref{thm:extraM}.} $\MSCLe\vdash\SCLe$.

\begin{proof}
With help of the theorem prover \emph{Prover9}~\cite{BirdBrain}.
We derive the \SCLe-axioms in a particular order,
as to obtain useful intermediate results. 
Recall that $(n)'$ represents the dual of equation $(n)$.

\medskip\noindent
\textbf{Axiom~\eqref{SCL3}.} First derive
\begin{equation}
\label{Aux19}
\tr\leftor x
\stackrel{\eqref{SCL4}}=\tr\leftand(\tr\leftor x)
\stackrel{\eqref{MSCL1}}=\tr.
\end{equation}
Hence,
\begin{align}
\nonumber
x
&=(\tr\leftor\tr)\leftand x
&&\text{by~\eqref{SCL4}, \eqref{Aux19}}\\
\nonumber
&=(\fa\leftand (\tr\leftand x))\leftor(\tr\leftand x)
&&\text{by~\eqref{Mem}, \eqref{SCL1}}\\
\label{20jan}
&=(\fa\leftand x)\leftor x,
&&\text{by~\eqref{SCL4}}
\end{align}
and
\begin{align}
\nonumber
\neg(\fa\leftand \neg x)
&=\neg(\neg\tr\leftand \neg x)
&&\text{by~\eqref{SCL1}}\\
\nonumber
&=\tr\leftor x
&&\text{by~\eqref{SCL2}}\\
\label{200jan}
&=\tr.
&&\text{by~\eqref{Aux19}}
\end{align}
Hence, $\neg(\fa\leftand x)
\stackrel{\eqref{20jan}}=\neg(\fa\leftand ((\fa\leftand x)\leftor x))
\stackrel{\eqref{SCL2}}=\neg(\fa\leftand \neg(\neg(\fa\leftand x)\leftand \neg x))
\stackrel{\eqref{200jan}}=\tr$,
and thus
\begin{align*}
z
&=(\tr\leftor y)\leftand z
&&\text{by~\eqref{SCL4}, \eqref{Aux19}}\\
&=(\fa\leftand(y\leftand z))\leftor (\tr\leftand z)
&&\text{by \eqref{Mem}, \eqref{SCL1}}\\
&=\neg(\neg(\fa\leftand(y\leftand z))\leftand \neg z)
&&\text{by~\eqref{SCL4}, \eqref{SCL2}}\\
&=\neg(\tr\leftand \neg z)
&&\text{by $\neg(\fa\leftand x)=\tr$}\\
\tag{\ref{SCL3}}
&=\neg\neg z.
&&\text{by \eqref{SCL4}}
\end{align*}

\medskip\noindent
\textbf{Intermediate result 1 - Duality.}
By axioms~\eqref{SCL1}-\eqref{SCL3} the duality principle holds.

\medskip\noindent
\textbf{Axiom~\eqref{SCL6}.} $\fa\leftand x=\fa$ by~$\eqref{Aux19}'$. 

\medskip\noindent
\textbf{Axiom~\eqref{SCL5}.} Instantiate~\eqref{Mem} with $x=\fa$ and $y=\tr$,
and apply $\neg\fa=\tr$ and~\eqref{SCL4}, \eqref{SCL6}:
\[
\tag{\ref{SCL5}}
z=\tr\leftand z
\stackrel{\eqref{SCL4}'}=(\fa\leftor \tr)\leftand z
\stackrel{\eqref{MSCL3}}=(\tr\leftand(\tr\leftand z))\leftor (\fa\leftand z)
\stackrel{\eqref{SCL4},\eqref{SCL6}}=z\leftor \fa.
\]

\medskip\noindent
\textbf{Intermediate result 2 - Idempotence.} 
By axiom~\eqref{SCL5},
$x=x\leftand(x\leftor\fa)\stackrel{\eqref{MSCL1}}=x\leftand x$.

\medskip\noindent
\textbf{Axiom~\eqref{SCL8}.} We derive the dual equation.
First derive 
\begin{align}
\nonumber
x\leftor \tr
&=(x\leftor \tr)\leftand\tr
&&\text{by~\eqref{SCL4}}\\
\nonumber
&=(\neg x\leftand(\tr\leftand\tr))\leftor(x\leftand\tr)
&&\text{by~\eqref{Mem}}\\
\label{hulpje2}
&=\neg x\leftor x, 
&&\text{by~\eqref{SCL4}}
\end{align}
and  
\begin{align}
\label{133dec}
&\neg x\leftor \tr=x\leftor \neg x.
&&\text{by~\eqref{hulpje2}, \eqref{SCL3}}
\end{align}
Hence
\begin{align*}
x\leftor\tr
&=(\neg x\leftor x)\leftand\tr
&&\text{by \eqref{hulpje2}, \eqref{SCL4}}\\
&=(x\leftand (x\leftand\tr))\leftor (\neg x\leftand\tr)
&&\text{by~\eqref{Mem}}\\
&=x\leftor\neg x
&&\text{by~\eqref{SCL4}, idempotence}\\
\tag*{\eqref{SCL8}$'$}
&=\neg x\leftor \tr.
&&\text{by~\eqref{133dec}}
\end{align*}

\medskip\noindent
\textbf{Intermediate result 3 - four auxiliary results.} 
\begin{align}
\nonumber
x\leftand y
&=(x\leftor\fa)\leftand y
&&\text{by~\eqref{SCL5}}\\
\nonumber
&=(\neg x\leftand(\fa\leftand y))\leftor (x\leftand y)
&&\text{by~\eqref{Mem}}\\
\label{ar 2}
\tag{Ar1}
&=(x\leftand \fa)\leftor (x\leftand y).
&&\text{by~\eqref{SCL6}, \eqref{SCL8}}
\\[2mm]
\nonumber
x\leftor y
&=(x\leftor y)\leftand\tr
&&\text{by \eqref{SCL5}$'$}\\
\label{ar 1}
\tag{Ar2}
&=(\neg x\leftand y)\leftor x.
&&\text{by \eqref{Mem}}
\\[2mm]
\nonumber
x\leftor y
&=(x\leftor \tr)\leftand(x\leftor y)
&&\text{by~\eqref{ar 2}$'$}\\
\nonumber
&=(\neg x\leftand(x\leftor y))\leftor(x\leftand(x\leftor y))
&&\text{by~\eqref{Mem}}\\
\nonumber
&=(\neg x\leftand(x\leftor y))\leftor x
&&\text{by~\eqref{MSCL1}}\\
\label{ar 3}
\tag{Ar3}
&=x\leftor(x\leftor y).
&&\text{by~\eqref{ar 1}}
\\[2mm]
\nonumber
x\leftor y
&=(\neg x\leftand y)\leftor x
&&\text{by~\eqref{ar 1}}\\
\nonumber
&=(\neg x\leftand (\neg x\leftand y))\leftor x
&&\text{by~\eqref{ar 3}$'$}\\
\label{ar 4}
\tag{Ar4}
&=x\leftor (\neg x\leftand y).
&&\text{by~\eqref{ar 1}}
\end{align}

\medskip\noindent
\textbf{Axiom~\eqref{SCL9}.} 
First derive
\begin{align}
\nonumber
(x\leftor\tr)\leftand\fa
&=(\neg x\leftand\fa)\leftor(x\leftand\fa)
&&\text{by~\eqref{Mem}, \eqref{SCL4}}\\
\nonumber
&=(x\leftand\fa)\leftor(x\leftand\fa)
&&\text{by~\eqref{SCL8}}\\
\label{19dec}
&=x\leftand\fa.
&&\text{by~idempotence}
\end{align}
Hence,
\begin{align*}
(x\leftor\tr)\leftand y
&=((x\leftor \tr)\leftand \fa)\leftor ((x\leftor\tr)\leftand y)
&&\text{by~\eqref{ar 2}}\\
&=(x\leftand \fa)\leftor ((x\leftor\tr)\leftand y)
&& \text{by~\eqref{19dec}}\\
&=(x\leftand \fa)\leftor (\neg(x\leftand\fa)\leftand y)
&& \text{by~\eqref{SCL8}$'$}\\
\tag{\ref{SCL9}}
&=(x\leftand\fa)\leftor y.
&& \text{by~\eqref{ar 4}}
\end{align*}

\medskip\noindent
\textbf{Intermediate result 4 - three more auxiliary results.}
\begin{align*}
(x\leftand \fa)\leftand y&=x\leftand \fa
\tag{\ref{ar 5}}\\
x\leftand(y\leftand x)&=x\leftand y
\tag{\ref{ar 6}}\\
(x\leftand y)\leftand x&=x\leftand y
\tag{\ref{ar 7}}
\end{align*}
First derive
\begin{align}
\nonumber
(x \leftand\fa)\leftand\fa
&=\neg(x \leftand\fa)\leftand\fa
&&\text{by \eqref{SCL8}}\\
\nonumber 
&=(\neg x\leftor\tr)\leftand\fa\\
\nonumber 
&=\neg x \leftand\fa
&&\text{by~\eqref{19dec}}\\
\label{185dec}
&=x\leftand\fa,
&&\text{by~\eqref{SCL8}}
\end{align}
hence
\begin{align}
\nonumber
(x \leftand\fa) \leftand y
&=(x \leftand\fa) \leftand ((x \leftand\fa) \leftand y)
&&\text{by~\eqref{ar 3}$'$}\\
\nonumber
&=(x \leftand\fa)\leftand (((x\leftand\fa)\leftand\fa)\leftor((x \leftand\fa) \leftand y))
&&\text{by~\eqref{ar 2}}\\
\nonumber
&=(x \leftand\fa)\leftand ((x\leftand\fa)\leftor((x \leftand\fa) \leftand y))
&&\text{by~\eqref{185dec}}\\
\label{ar 5}
\tag{Ar5}
&=x \leftand\fa.
&&\text{by~\eqref{MSCL1}}
\\[2mm]
\nonumber
x\leftand (y\leftand x)
&=(x\leftand x)\leftand (y\leftand x)
&&\text{by idempotence}\\
\nonumber
&=((x\leftand\fa)\leftor x)\leftand (y\leftand x)
&&\text{by \eqref{ar 2}, idempotence}\\
\nonumber
&=(\neg (x\leftand\fa)\leftand (x\leftand (y\leftand x)))\leftor ((x\leftand \fa)\leftand (y\leftand x))
&&\text{by~\eqref{Mem}}\\
\nonumber
&=((\neg x\leftor\tr)\leftand (x\leftand (y\leftand x)))\leftor (x\leftand \fa)
&&\text{by~\eqref{ar 5}}\\
\nonumber
&=((x\leftand\fa)\leftor (x\leftand (y\leftand x)))\leftor (x\leftand \fa)
&&\text{by~\eqref{SCL9}, \eqref{SCL8}}\\
\nonumber
&=(x\leftand (y\leftand x))\leftor (x\leftand \fa)
&&\text{by~\eqref{ar 2}}\\
\nonumber
&=(x\leftand (y\leftand x))\leftor (\neg x\leftand x)
&&\text{by~\eqref{hulpje2}$'$}\\
\nonumber
&=(\neg x\leftor y)\leftand x
&&\text{by~\eqref{Mem}}\\
\label{ar 6}
\tag{Ar6}
&=x\leftand y.
&&\text{by \eqref{ar 1}$'$}
\\[2mm]
\nonumber
(x\leftand y)\leftand x
&=(x\leftand y)\leftand(x\leftand(x\leftand y))
&&\text{by \eqref{ar 6}}\\
\nonumber
&=(x\leftand y)\leftand (x\leftand y)
&&\text{by \eqref{ar 3}$'$}\\
\label{ar 7}
\tag{Ar7}
&=x\leftand y.
&&\text{by idempotence}
\end{align}

\medskip\noindent
\textbf{Axiom~\eqref{SCL7}.}
We use the following auxiliary results:
\begin{align*}
(x\leftand y)\leftand z&=(x\leftand\fa)\leftor ((x\leftand y)\leftand z)
&&\eqref{62}
\\
(x\leftor y)\leftand(y\leftand z)&=(x\leftand\fa)\leftor (y\leftand z)
&&\eqref{above 1}
\\
\neg x\leftor(y\leftand z)&=\neg x\leftor((x\leftand y)\leftand z)
&&\eqref{13784}
\end{align*}
and derive associativity of $\leftand$ as follows:
\begin{align*}
(x\leftand y)\leftand z
&=(x\leftand \fa)\leftor ((x\leftand y)\leftand z)
&&\text{by~\eqref{62}}\\
&=(x\leftor(x\leftand y))\leftand((x\leftand y)\leftand z)
&&\text{by~\eqref{above 1}}\\
&=x\leftand((x\leftand y)\leftand z)
&&\text{by~\eqref{MSCL1}$'$}\\
&=(\neg x \leftor((x\leftand y)\leftand z))\leftand x
&&\text{by~\eqref{ar 1}$'$}\\
&=(\neg x \leftor(y\leftand z))\leftand x
&&\text{by~\eqref{13784}}\\
\tag{\ref{SCL7}}
&=x\leftand (y\leftand z).
&&\text{by~\eqref{ar 1}$'$}
\end{align*}

We derive the above auxiliary results in order: 
\begin{align}
\nonumber
(x\leftand\fa)\leftor ((x\leftand y)\leftand z)
&=((x\leftand\fa)\leftor ((x\leftand y)\leftand z))\leftor(x\leftand\fa)
&&\text{by \eqref{ar 7}$'$}\\
\nonumber
&=((x\leftor\tr)\leftand ((x\leftand y)\leftand z))\leftor(x\leftand\fa)
&&\text{by~\eqref{SCL9}}\\
\nonumber
&=((\neg x\leftor\tr)\leftand ((x\leftand y)\leftand z))\leftor(x\leftand\fa)
&&\text{by~\eqref{SCL8}$'$}\\
\nonumber
&=(\neg(x\leftand\fa)\leftand ((x\leftand y)\leftand z))\leftor((x\leftand\fa)\leftand z)
&&\text{by \eqref{ar 5}}\\
\nonumber
&=((x\leftand\fa)\leftor (x\leftand y))\leftand z
&&\text{by~\eqref{Mem}}\\
\label{62}
&=(x\leftand y)\leftand z.
&&\text{by \eqref{ar 2}}
\end{align}
\begin{align}
\nonumber
(x\leftor y)\leftand(y\leftand z)
&=(\neg x\leftand(y\leftand(y\leftand z)))\leftor(x\leftand(y\leftand z))
&&\text{by~\eqref{Mem}}\\
\nonumber
&=(\neg x\leftand(y\leftand z))\leftor(x\leftand(y\leftand z))
&&\text{by \eqref{ar 3}$'$}\\
\nonumber
&=(x\leftor\tr)\leftand (y\leftand z)
&&\text{by~\eqref{Mem}}\\
\label{above 1}
&=(x\leftand\fa)\leftor (y\leftand z),
&&\text{by~\eqref{SCL9}}
\end{align}
The remaining auxiliary results lead to~\eqref{13784}:
\begin{align}
\nonumber
(x\leftor y)\leftand(x\leftor z)
&=(\neg x\leftand (y\leftand(x\leftor z)))\leftor
(x\leftand (x\leftor z))
&&\text{by \eqref{Mem}}\\
\nonumber
&=(\neg x\leftand (y\leftand(x\leftor z)))\leftor x
&&\text{by~\eqref{MSCL1}}\\
\label{39}
&=x\leftor(y\leftand(x\leftor z)).
&&\text{by \eqref{ar 1}}
\\[2mm]
\nonumber
x\leftor (\neg x\leftor y)
&=(\neg x\leftand(\neg x\leftor y))\leftor x
&&\text{by \eqref{ar 1}}\\
\nonumber
&=\neg x\leftor x
&&\text{by \eqref{MSCL1}}\\
\nonumber
&=(\neg x\leftand\neg x)\leftor x
&&\text{by idempotence}\\
\label{1701}
&=x\leftor \neg x.
&&\text{by \eqref{ar 1}}
\\[2mm]
\nonumber
x\leftor((x\leftand z)\leftor y)
&=x\leftor((\neg x\leftor(z\leftor y))\leftand(x\leftor y))
&&\text{by \eqref{Mem}}\\
\nonumber
&=(x\leftor(\neg x\leftor(z\leftor y)))\leftand(x\leftor y)
&&\text{by \eqref{39}}\\
\nonumber
&=(x\leftor\neg x)\leftand(x\leftor y)
&&\text{by \eqref{1701}}\\
\nonumber
&=x\leftor(\neg x\leftand(x\leftor y))
&&\text{by \eqref{39}}\\
\nonumber
&=x\leftor(x\leftor y)
&&\text{by \eqref{ar 4}}\\
\label{997}
&=x\leftor y.
&&\text{by \eqref{ar 3}}
\\[2mm]
\nonumber
x\leftor y
&=x\leftor((x\leftand z)\leftor y)
&&\text{by \eqref{997}}\\
\nonumber
&=x\leftor((x\leftand z)\leftor (y\leftor(x\leftand z)))
&&\text{by \eqref{ar 6}$'$}\\
\label{6048}
&=x\leftor (y\leftor(x\leftand z)).
&&\text{by \eqref{997}}
\\[2mm]
\nonumber
x\leftor(y\leftand z)
&=x\leftor(\neg x\leftand (y\leftand z))
&&\text{by \eqref{ar 4}}\\
\nonumber
&=x\leftor((\neg x\leftand (y\leftand z))\leftor(x\leftand z))
&&\text{by \eqref{6048}}\\
\label{13359}
&=x\leftor((x\leftor y)\leftand z).
&&\text{by~\eqref{Mem}}
\\[2mm]
\nonumber
\neg x\leftor(y\leftand z)
&=\neg x\leftor((\neg x\leftor y)\leftand z)
&&\text{by \eqref{13359}}\\
\nonumber
&=\neg x\leftor((\neg x\leftor (x\leftand y))\leftand z)
&&\text{by \eqref{ar 4}}\\
\label{13784}
&=\neg x\leftor((x\leftand y)\leftand z).
&&\text{by \eqref{13359}}
\end{align}
We write ``Assoc'' for (repeated) applications of associativity
of $\leftand$ and $\leftor$. 

\newpage
\noindent
\textbf{Intermediate result 5.}  
In order to derive axiom~\eqref{SCL10} we use the following two intermediate results:
\begin{align*}
(x\leftand y)\leftor(\neg x\leftand z)&=(\neg x\leftor y)\leftand(x\leftor z)
\tag{\ref{M1}}\\
(x\leftand y)\leftor(\neg x\leftand z)&=(\neg x\leftand z)\leftor(x\leftand y)
\tag{\ref{M2}}
\end{align*}
Equation \eqref{M1}.
\begin{align}
\nonumber
(x\leftand y)\leftor(\neg x\leftand z)
&=(x\leftand y)\leftor(\neg x\leftand(x\leftor z))
&&\text{by (ir3)$'$}\\
\nonumber
&=(x\leftand(y\leftand(x\leftor z)))\leftor(\neg x\leftand(x\leftor z))
&&\text{by~\eqref{6048}$'$}\\
\label{M1}
\tag{M1}
&=(\neg x\leftor y)\leftand(x\leftor z).
&&\text{by~\eqref{Mem}}
\end{align}

\noindent
Equation \eqref{M2}.
First derive
\begin{align}
\nonumber
(x\leftor y)\leftand z
&=(x\leftor y)\leftand ((x\leftor y)\leftand z)
&&\text{by \eqref{ar 3}$'$}\\
\nonumber
&=(x\leftor (x\leftor y))\leftand ((x\leftor y)\leftand z)
&&\text{by \eqref{ar 3}}\\
\label{638}
&=(x\leftand\fa)\leftor ((x\leftor y)\leftand z),
&&\text{by~\eqref{above 1}}
\\[2mm]
\nonumber
(x\leftand \fa)\leftor y
&=(\neg x\leftand \fa)\leftor y
&&\text{by~\eqref{SCL8}}\\
\nonumber
&=(\neg x\leftor\tr)\leftand y
&&\text{by~\eqref{SCL9}}\\
\nonumber
&=(x\leftand y)\leftor(\neg x\leftand y)
&&\text{by~\eqref{Mem}, \eqref{SCL4}}\\
\nonumber 
&=(\neg x\leftor y)\leftand(x\leftor y)
&&\text{by~\eqref{M1}}\\
\nonumber 
&=(\neg x\leftor (y\leftor y))\leftand(x\leftor y)
&&\text{by~idempotence}\\
\label{3296}
&=(x\leftand y)\leftor y,
&&\text{by~\eqref{Mem}$'$}
\\[2mm]
\nonumber
(x\leftor y)\leftand z
&=(x\leftand\fa)\leftor((x\leftor y)\leftand z)
&&\text{by~\eqref{638}}\\
\nonumber
&=(x\leftand((x\leftor y)\leftand z))\leftor((x\leftor y)\leftand z)
&&\text{by~\eqref{3296}}\\
\nonumber
&=((x\leftand (x\leftor y))\leftand z)\leftor((x\leftor y)\leftand z)
&&\text{by Assoc}\\
\label{3381}
&=(x\leftand z)\leftor((x\leftor y)\leftand z).
&&\text{by \eqref{MSCL1}}
\end{align}
Hence,
\begin{align}
\nonumber
(x\leftand y)\leftor(\neg x\leftand z)
&=(\neg x\leftor y)\leftand(x\leftor z)
&&\text{by~\eqref{M1}}\\
\nonumber
&=(\neg x\leftor y)\leftand((x\leftor z)\leftand(\neg x\leftor y))
&&\text{by~\eqref{ar 6}}\\
\nonumber
&=(\neg x\leftor (x\leftand y))\leftand((x\leftor z)\leftand(\neg x\leftor y))
&&\text{by~\eqref{ar 4}}\\
\nonumber
&=(\neg x\leftor (x\leftand y))\leftand((\neg x\leftand z)\leftor(x\leftand y))
&&\text{by~\eqref{M1}$'$}\\
\label{M2}
\tag{M2}
&=(\neg x\leftand z)\leftor(x\leftand y).
&&\text{by~\eqref{3381}$'$}
\end{align}

\medskip\noindent
\textbf{Axiom~\eqref{SCL10}.}
First derive
\begin{align}
\nonumber
(x\leftor y)\leftand z
&=(\neg x\leftand (y\leftand z))\leftor(x\leftand z)
&&\text{by~\eqref{Mem}}\\
\label{eq:8}
&=(x\leftor(y\leftand z))\leftand (\neg x\leftor z).
&&\text{by \eqref{M1}, \eqref{SCL3}}
\end{align}
Hence,
\begin{align*}
(x\leftand y)\leftor(z\leftand\fa)
&=
(x\leftand(y\leftor(z\leftand\fa)))\leftor(\neg x\leftand(z\leftand\fa))
&& \text{by~\eqref{eq:8}$'$}\\
&=(x\leftor(z\leftand\fa))\leftand
(\neg x\leftor (y\leftor(z\leftand\fa)))
&&\text{by \eqref{M1}, \eqref{M2}}\\
&=(x\leftor[(z\leftand\fa)\leftand(y\leftor(z\leftand\fa))])
\leftand
(\neg x\leftor(y\leftor(z\leftand\fa)))
&&\text{by \eqref{ar 5}}\\
\tag{\ref{SCL10}}
&=(x\leftor (z\leftand\fa))\leftand(y\leftor (z\leftand\fa)).
&& \text{by~\eqref{eq:8}}
\end{align*}
\end{proof}
\newpage
\subsection{A proof of Theorem~\ref{thm:dist}}
\label{app:LD}

\textbf{Theorem~\ref{thm:dist}.} 
The following equations are derivable from \MSCLe, where~\eqref{D1} abbreviates left-distributivity.
\begin{align}
\tag{\ref{M3}}
&((x\leftand y)\leftor(\neg x\leftand z))\leftand u
=(x\leftand (y\leftand u))\leftor(\neg x\leftand (z\leftand u)),
\\
\tag{\ref{D1}}
&x\leftand(y\leftor z)=(x\leftand y)\leftor(x\leftand z).
\end{align}

\smallskip
\begin{proof}
With help of the theorem prover \emph{Prover9}~\cite{BirdBrain}.

\medskip\noindent
\textbf{Equation \eqref{M3}.}
First derive
\begin{align}
\nonumber
x\leftand(y\leftand((x\leftor z)\leftand u))
&=(x\leftand (y\leftand(x\leftor z)))\leftand u
&&\text{by Assoc}\\
\nonumber
&=(x\leftand y)\leftand u
&&\text{by~\eqref{6048}$'$}\\
\label{above n3}
&=x\leftand (y\leftand u).
&&\text{by Assoc}
\end{align}
Hence,
\begin{align}
\nonumber
((x\leftand y)\leftor(\neg x\leftand z))\leftand u
&=((\neg x\leftor y)\leftand(x\leftor z))\leftand u
&&\text{by~\eqref{M1}}\\
\nonumber
&=(\neg x\leftor y)\leftand((x\leftor z)\leftand u)
&&\text{by Assoc}\\
\nonumber
&=(x\leftand (y\leftand((x\leftor z)\leftand u)))~\leftor\\
\nonumber
&\phantom{=~}(\neg x\leftand((x\leftor z))\leftand u)
&&\text{by~\eqref{Mem}}\\
\nonumber
&=(x\leftand (y\leftand u))\leftor(\neg x\leftand((x\leftor z))\leftand u)
&&\text{by~\eqref{above n3}}\\
\nonumber
&=(x\leftand (y\leftand u))\leftor((\neg x\leftand(x\leftor z))\leftand u)
&&\text{by Assoc}\\
\nonumber
&=(x\leftand (y\leftand u))\leftor((\neg x\leftand z)\leftand u)
&&\text{by \eqref{ar 4}$'$}\\
\tag{\ref{M3}}
&=(x\leftand (y\leftand u))\leftor(\neg x\leftand (z\leftand u)).
&&\text{by Assoc}
\end{align}

\medskip\noindent 
\textbf{Equation~\eqref{D1}.}
First derive
\begin{align}
\nonumber
x\leftor(y\leftor z)
&=(x\leftor y)\leftor z
&&\text{by Assoc}\\
\nonumber
&=((x\leftor y)\leftor x)\leftor z
&&\text{by~\eqref{ar 7}$'$}\\
\label{94}
&=x\leftor (y\leftor (x\leftor z)),
&&\text{by Assoc}
\\[2mm]
\nonumber
\neg x\leftor(y\leftor (x\leftand z))
&=\neg x\leftor(y\leftor (\neg x\leftor (x\leftand z)))
&&\text{by~\eqref{94}}\\
\nonumber
&=\neg x\leftor(y\leftor (\neg x\leftor z))
&&\text{by~\eqref{ar 4}}\\
\label{1660}
&=\neg x\leftor(y\leftor z).
&&\text{by~\eqref{94}}
\end{align}
Hence,
\begin{align}
\nonumber
x\leftand(y\leftor z)
&=(\neg x\leftor(y\leftor z))\leftand x
&&\text{by~\eqref{ar 1}$'$}\\
\nonumber
&=(\neg x\leftor(y\leftor z))\leftand (x\leftor(x\leftand z))
&&\text{by~\eqref{MSCL1}$'$}\\
\nonumber
&=(\neg x\leftor(y\leftor (x\leftand z)))\leftand (x\leftor(x\leftand z))
&&\text{by~\eqref{1660}}\\
\tag{\ref{D1}}
&=(x\leftand y)\leftor(x\leftand z).
&&\text{by~\eqref{Mem}$'$}
\end{align}
\end{proof}

\newpage

\subsection{A proof of Theorem~\ref{thm:assoc}}
\label{app:assoc}

\textbf{Theorem~\ref{thm:assoc}.}
The four \SSCLe-axioms~\eqref{Or}, \eqref{MSCL1}, \eqref{Mem}, and~\eqref{C1}
imply idempotence and associativity of ${\leftand}$ and $\leftor$, 
the double negation shift $\neg\neg x=x$ (that is, axiom~\eqref{SCL3}), and  
the equations
\begin{align}
\tag{\ref{Tdef}}
(x\leftor\neg x)\leftand y&=y,\\
\nonumber
\tag{\ref{D1}}
x\leftand(y\leftor z)
&=(x\leftand y)\leftor(x\leftand z).
\end{align}
Furthermore, if $|A|\geq 2$, these four axioms are independent.

\smallskip

\begin{proof}
With help of the theorem prover \emph{Prover9}~\cite{BirdBrain}.
Observe that $x\leftor y=y\leftor x$ readily follows from
the axioms~\eqref{Or} and~\eqref{C1}; we refer
to this equation by~\eqref{C1}$'$. 

\medskip\noindent
\textbf{\small Idempotence of $\leftand$.} We first derive
\begin{align}
\nonumber
(y\leftand x)\leftand((y\leftor z)\leftand x)
&=(y\leftand x)\leftand((y\leftand x)\leftor(\neg y\leftand (z\leftand x)))
&&\text{by~\eqref{Mem}, \eqref{C1}$'$}\\
\nonumber
&=y\leftand x,
&&\text{by~\eqref{MSCL1}}
\end{align}
hence
\begin{align}
\label{AHAHA}
(x\leftand y)\leftand((y\leftor z)\leftand x)
&=y\leftand x.
&&\text{by the above and \eqref{C1}}
\end{align}
Finally, 
\begin{align*}
x
&=(x\leftor y)\leftand x
&&\text{by~\eqref{MSCL1}, \eqref{C1}}\\
&=(x\leftand (x\leftor y))\leftand(((x\leftor y)\leftor x)\leftand x)
&&\text{by~\eqref{AHAHA} (substitute $x\leftor y$ for $y$, and $x$ for $z$)}\\
&=x\leftand(((x\leftor y)\leftor x)\leftand x)
&&\text{by~\eqref{MSCL1}}\\
&=x\leftand(x\leftand (x\leftor (x\leftor y)))
&&\text{by~\eqref{C1}, \eqref{C1}$'$}\\
&=x\leftand x.
&&\text{by~\eqref{MSCL1}}
\end{align*}

\noindent
\textbf{Idempotence of $\leftor$.} We first derive three auxiliary results:
\begin{align}
\nonumber
x\leftor y
&=(x\leftor y)\leftand (x\leftor y)\\
\nonumber
&=(x\leftand(x\leftor y))\leftor(\neg x\leftand(y\leftand(x\leftor y)))
&&\text{by~\eqref{Mem}, \eqref{C1}$'$}\\
\nonumber
&=x\leftor(\neg x\leftand(y\leftand(x\leftor y)))
&&\text{by~\eqref{MSCL1}}\\
\label{ACA}
&=x\leftor(\neg x\leftand y),
&&\text{by~\eqref{C1}$'$, \eqref{MSCL1}}
\\[2mm]
\nonumber
x
&=(x\leftor y)\leftand x
&&\text{by~\eqref{MSCL1}, \eqref{C1}}\\
\nonumber
&=(\neg x\leftand(y\leftand x))\leftor(x\leftand x)
&&\text{by~\eqref{Mem}}\\
\nonumber
&=x\leftor(\neg x\leftand(y\leftand x))
&&\text{by~idempotence of $\leftand$ and \eqref{C1}$'$}\\
\nonumber
&=x\leftor(y\leftand x)
&&\text{by~\eqref{ACA}}\\
\label{ADA}
&=x\leftor(x\leftand y),
&&\text{by~\eqref{C1}}
\\[2mm]
\nonumber
(x\leftor y)\leftand(x\leftor z)
&=(\neg x\leftand(y\leftand(x\leftor z)))\leftor (x\leftand(x\leftor z))
&&\text{by~\eqref{Mem}}\\
\nonumber
&=x\leftor (\neg x\leftand(y\leftand(x\leftor z)))
&&\text{by~\eqref{MSCL1}, \eqref{C1}$'$}\\
\label{AEA}
&=x\leftor (y\leftand(x\leftor z)).
&&\text{by~\eqref{ACA}}
\end{align}
Hence,
\begin{align*}
x\leftor x
&=(x\leftor x)\leftand(x\leftor x)
\\
&=x\leftor(x\leftand(x\leftor x))
&&\text{by~\eqref{AEA}}\\
&=x.
&&\text{by~\eqref{ADA}}\\[-2mm]
\end{align*}

\medskip\noindent
\textbf{Double negation shift.}
With idempotence, the double negation shift follows immediately:
\[
\neg\neg x=\neg(\neg x\leftand \neg x)\stackrel{\eqref{Or}}=x\leftor x=x.
\]

Hence the duality principle holds in the setting without \tr\ and \fa,
which justifies the name \eqref{C1}$'$ for the equation $x\leftor y=y\leftor x$.

\medskip\noindent
\textbf{Equation~\eqref{Tdef}}, that is $(x\leftor \neg x)\leftand y=y$.
First derive
\begin{align}
\nonumber
x\leftand(x\leftand y)
&=(x\leftand y)\leftand x
&&\text{by~\eqref{C1}}\\
\nonumber
&=(x\leftand y)\leftand(x\leftor(x\leftand y))
&&\text{by~\eqref{ADA}}\\
\nonumber
&=(x\leftand y)\leftand((x\leftand y)\leftor x)
&&\text{by~\eqref{C1}$'$}\\
\label{AFA}
&=x\leftand y,
&&\text{by~\eqref{MSCL1}}
\end{align}
hence
\begin{align}
\nonumber
(x\leftor\neg x)\leftand y
&=(\neg x\leftand(\neg x\leftand y))\leftor(x\leftand y)
&&\text{by~\eqref{Mem}}\\
\nonumber
&=(\neg x\leftand y)\leftor(x\leftand y)
&&\text{by~\eqref{AFA}}\\
\nonumber
&=(\neg x\leftand (y\leftand y))\leftor(x\leftand y)
&&\text{by idempotence}\\
\nonumber
&=(x\leftor y)\leftand y
&&\text{by~\eqref{Mem}}\\
\tag{\ref{Tdef}}
&=y.
&&\text{by~\eqref{C1}, \eqref{C1}$'$, \eqref{MSCL1}}
\end{align}

\medskip

For the remaining statements of the theorem, that is, 
associativity and  left-distributivity~\eqref{D1},
we can refer to the associated \MSCLe-derivations: by duality,
equation~\eqref{Tdef},
and the observation that in each \MSCLe-derivation, the constant \tr\ can be represented by 
$u\leftor\neg u$ with $u$ a fresh variable, the 
counterparts of the axioms~\eqref{Tand} and~\eqref{Neg} are available, and therefore
the \MSCLe-derivations of these equations can be adapted in this way.

\medskip

The independence of the four \SSCLe-axioms~\eqref{Or}, \eqref{MSCL1}, \eqref{Mem}, and~\eqref{C1}
requires that $|A|\geq 2$ and is proved in Appendix~\ref{app:Indep2}.
\end{proof}

\newpage

\subsection{A proof of Theorem~\ref{thm:Indep2}}
\label{app:Indep2}
\textbf{Theorem~\ref{thm:Indep2}.} The axioms of $\SSCLe$ are independent.

\begin{proof}
All independence models were found with  
the tool \emph{Mace4}~\cite{BirdBrain}.
In each model \M\ defined below, $\llbracket \fa\rrbracket^\M=0$ and 
$\llbracket \tr\rrbracket^\M=1$.
Recall that $A\ne\emptyset$ and observe that one atom $a$ is used to show the independence of 
axioms~\eqref{Or}, \eqref{Mem}, and~\eqref{C1}.
The independence result stated in Theorem~\ref{thm:assoc} follows 
by using in the refutations below two atoms 
instead of \fa\ and \tr, so this result requires that $|A|\geq 2$.

Independence of axiom~\eqref{Neg}.
A model \M\ for $\SSCLe\setminus\{\eqref{Neg}\}$
with domain $\{0,1,2,3\}$ that refutes $\fa=\neg\tr$ is the following:
\[
\begin{array}{r@{\hspace{5pt}}|@{\hspace{5pt}}c}
\neg\\\hline\\[-3mm]
0&3\\
1&2\\
2&1\\
3&0
\end{array}
\qquad\qquad
\begin{array}{r@{\hspace{5pt}}|@{\hspace{5pt}}c@{\hspace{5pt}}c@{\hspace{5pt}}c@{\hspace{5pt}}c@{\hspace{5pt}}c}
\leftand&
0&1&2&3\\\hline\\[-3mm]
0&
0&0&2&2\\
1&
0&1&2&3\\
2&
2&2&2&2\\
3&
2&3&2&3
\end{array}
\qquad\qquad
\begin{array}{r@{\hspace{5pt}}|@{\hspace{5pt}}c@{\hspace{5pt}}c@{\hspace{5pt}}c@{\hspace{5pt}}c@{\hspace{5pt}}c}
\leftor&
0&1&2&3\\\hline\\[-3mm]
0&
0&1&0&1\\
1&
1&1&1&1\\
2&
0&1&2&3\\
3&
1&1&3&3
\end{array}
\] 

 Independence of axiom~\eqref{Or}.
A model \M\ for $\SSCLe\setminus\{\eqref{Or}\}$
with domain $\{0,1,2\}$ and $\llbracket a\rrbracket^\M=2$ for some $a\in A$
that refutes $\fa\leftor a=\neg(\neg\fa\leftand \neg a)$ is the following:
\[
\begin{array}{r@{\hspace{5pt}}|@{\hspace{5pt}}c}
\neg\\\hline\\[-3mm]
0&1\\
1&0\\
2&0
\end{array}
\qquad\qquad
\begin{array}{r@{\hspace{5pt}}|@{\hspace{5pt}}c@{\hspace{5pt}}c@{\hspace{5pt}}c@{\hspace{5pt}}c@{\hspace{5pt}}c}
\leftand&
0&1&2\\\hline\\[-3mm]
0&
0&0&0\\
1&
0&1&2\\
2&
0&2&2
\end{array}
\qquad\qquad
\begin{array}{r@{\hspace{5pt}}|@{\hspace{5pt}}c@{\hspace{5pt}}c@{\hspace{5pt}}c@{\hspace{5pt}}c@{\hspace{5pt}}c}
\leftor&
0&1&2\\\hline\\[-3mm]
0&
0&1&2\\
1&
1&1&1\\
2&
2&1&2
\end{array}
\] 

 Independence of axiom~\eqref{Tand}.
A model \M\ for $\SSCLe\setminus\{\eqref{Tand}\}$
with domain $\{0,1\}$ that refutes $\tr\leftand\fa=\fa$ is the following:
\[
\begin{array}{r@{\hspace{5pt}}|@{\hspace{5pt}}c}
\neg\\\hline\\[-3mm]
0&1\\
1&0
\end{array}
\qquad\qquad
\begin{array}{r@{\hspace{5pt}}|@{\hspace{5pt}}c@{\hspace{5pt}}c@{\hspace{5pt}}c@{\hspace{5pt}}c@{\hspace{5pt}}c}
\leftand&
0&1\\\hline\\[-3mm]
0&
0&1\\
1&
1&1
\end{array}
\qquad\qquad
\begin{array}{r@{\hspace{5pt}}|@{\hspace{5pt}}c@{\hspace{5pt}}c@{\hspace{5pt}}c@{\hspace{5pt}}c@{\hspace{5pt}}c}
\leftor&
0&1\\\hline\\[-3mm]
0&
0&0\\
1&
0&1
\end{array}
\] 

 Independence of axiom~\eqref{MSCL1}.
A model \M\ for $\SSCLe\setminus\{\eqref{MSCL1}\}$
with domain $\{0,1\}$ that refutes
\\
$\tr\leftand(\tr\leftor\fa)=\tr$ is the following:
\[
\begin{array}{r@{\hspace{5pt}}|@{\hspace{5pt}}c}
\neg\\\hline\\[-3mm]
0&0\\
1&0
\end{array}
\qquad\qquad
\begin{array}{r@{\hspace{5pt}}|@{\hspace{5pt}}c@{\hspace{5pt}}c@{\hspace{5pt}}c@{\hspace{5pt}}c@{\hspace{5pt}}c}
\leftand&
0&1\\\hline\\[-3mm]
0&
0&0\\
1&
0&1
\end{array}
\qquad\qquad
\begin{array}{r@{\hspace{5pt}}|@{\hspace{5pt}}c@{\hspace{5pt}}c@{\hspace{5pt}}c@{\hspace{5pt}}c@{\hspace{5pt}}c}
\leftor&
0&1\\\hline\\[-3mm]
0&
0&0\\
1&
0&0
\end{array}
\] 

 Independence of axiom~\eqref{Mem}.
A model \M\ for $\SSCLe\setminus\{\eqref{MSCL3}\}$
with domain $\{0,1,2\}$ and $\llbracket a\rrbracket^\M=2$ for some $a\in A$
that refutes $(\fa \leftor \fa) \leftand a
=(\neg\fa \leftand (\fa \leftand a)) \leftor (\fa \leftand a)$ 
is the following:
\[
\begin{array}{r@{\hspace{5pt}}|@{\hspace{5pt}}c}
\neg\\\hline\\[-3mm]
0&1\\
1&0\\
2&0
\end{array}
\qquad\qquad
\begin{array}{r@{\hspace{5pt}}|@{\hspace{5pt}}c@{\hspace{5pt}}c@{\hspace{5pt}}c@{\hspace{5pt}}c@{\hspace{5pt}}c}
\leftand&
0&1&2\\\hline\\[-3mm]
0&
0&0&2\\
1&
0&1&2\\
2&
2&2&0
\end{array}
\qquad\qquad
\begin{array}{r@{\hspace{5pt}}|@{\hspace{5pt}}c@{\hspace{5pt}}c@{\hspace{5pt}}c@{\hspace{5pt}}c@{\hspace{5pt}}c}
\leftor&
0&1&2\\\hline\\[-3mm]
0&
0&1&1\\
1&
1&1&1\\
2&
1&1&1
\end{array}
\] 

 Independence of axiom~\eqref{C1}.
A model \M\ for $\SSCLe\setminus\{\eqref{C1}\}$
with domain $\{0,1,2\}$ and $\llbracket a\rrbracket^\M=2$ for some $a\in A$
that refutes $\fa \leftand a
=a \leftand \fa$ 
is the following:
\[
\begin{array}{r@{\hspace{5pt}}|@{\hspace{5pt}}c}
\neg\\\hline\\[-3mm]
0&1\\
1&0\\
2&2
\end{array}
\qquad\qquad
\begin{array}{r@{\hspace{5pt}}|@{\hspace{5pt}}c@{\hspace{5pt}}c@{\hspace{5pt}}c@{\hspace{5pt}}c@{\hspace{5pt}}c}
\leftand&
0&1&2\\\hline\\[-3mm]
0&
0&0&0\\
1&
0&1&2\\
2&
2&2&2
\end{array}
\qquad\qquad
\begin{array}{r@{\hspace{5pt}}|@{\hspace{5pt}}c@{\hspace{5pt}}c@{\hspace{5pt}}c@{\hspace{5pt}}c@{\hspace{5pt}}c}
\leftor&
0&1&2\\\hline\\[-3mm]
0&
0&1&2\\
1&
1&1&1\\
2&
2&2&2
\end{array}
\] 

\end{proof}
\end{document}